%% file: main.tex
\documentclass[a4paper,UKenglish,cleveref,autoref,thm-restate]{lipics-v2021}
\title{Unifying Sequent Systems for G{\"o}del-L{\"o}b Provability Logic via Syntactic Transformations} %TODO Please add
\titlerunning{Unifying Sequent Systems for G{\"o}del-L{\"o}b Provability Logic} %TODO optional, please use if title is longer than one line

\author{Tim S. Lyon}{Technische Universit{\"a}t Dresden, Germany \and \url{https://sites.google.com/view/timlyon} }{timothy_stephen.lyon@tu-dresden.de}{https://orcid.org/0000-0003-3214-0828}{}%TODO mandatory, please use full name; only 1 author per \author macro; first two parameters are mandatory, other parameters can be empty. Please provide at least the name of the affiliation and the country. The full address is optional. Use additional curly braces to indicate the correct name splitting when the last name consists of multiple name parts.

\authorrunning{T.S. Lyon} %TODO mandatory. First: Use abbreviated first/middle names. Second (only in severe cases): Use first author plus 'et al.'

\Copyright{Timothy Stephen Lyon} %TODO mandatory, please use full first names. LIPIcs license is "CC-BY";  http://creativecommons.org/licenses/by/3.0/


 %TODO mandatory; please add comma-separated list of keywords

\category{} %optional, e.g. invited paper

\relatedversion{} %optional, e.g. full version hosted on arXiv, HAL, or other respository/website
\funding{Work supported by European Research Council, Consolidator Grant DeciGUT (771779).}%optional, to capture a funding statement, which applies to all authors. Please enter author specific funding statements as fifth argument of the \author macro.

\nolinenumbers %uncomment to disable line numbering

\EventEditors{J\"{o}rg Endrullis and Sylvain Schmitz}
\EventNoEds{2}
\EventLongTitle{33rd EACSL Annual Conference on Computer Science Logic (CSL 2025)}
\EventShortTitle{CSL 2025}
\EventAcronym{CSL}
\EventYear{2025}
\EventDate{February 10--14, 2025}
\EventLocation{Amsterdam, Netherlands}
\EventLogo{}
\SeriesVolume{326}
\ArticleNo{42}
\input{commands.tex}

\begin{document}

\maketitle

\begin{abstract}
We demonstrate the inter-translatability of proofs between the most prominent sequent-based formalisms for Gödel-Löb provability logic. In particular, we consider Sambin and Valentini's sequent system $\seqgl$, Shamkanov's non-wellfounded and cyclic sequent systems $\seqgli$ and $\seqglc$, Poggiolesi's tree-hypersequent system $\treegl$, and Negri's labeled sequent system $\gtgl$. Shamkanov provided proof-theoretic correspondences between $\seqgl$, $\seqgli$, and $\seqglc$, and Gor{\'e} and Ramanayake showed how proofs could be transformed between $\treegl$ and $\gtgl$; however, the exact nature of proof transformations between the former three systems and the latter two systems has remained an open problem.  We solve this open problem by showing how to restructure tree-hypersequent proofs into an end-active form and introduce a novel \emph{linearization technique} that transforms such proofs into linear nested sequent proofs. As a result, we obtain a new proof-theoretic tool for extracting linear nested sequent systems from tree-hypersequent systems, which yields the first cut-free linear nested sequent calculus $\lingl$ for Gödel-Löb provability logic. We show how to transform proofs in $\lingl$ into a certain normal form, where proofs repeat in stages of modal and local rule applications, and which are translatable into $\seqgl$ and $\gtgl$ proofs. These new syntactic transformations, together with those mentioned above, establish full proof-theoretic correspondences between $\seqgl$, $\seqgli$, $\seqglc$, $\treegl$, $\gtgl$, and $\lingl$ while also giving (to the best of the author's knowledge) the first constructive proof mappings between structural (viz. labeled, tree-hypersequent, and linear nested sequent) systems and a cyclic sequent system.
\end{abstract}

%%NOTE: Used "appref" to mark references to the appendix

\section{Introduction}\label{sec:intro}

\input{intro.tex}

\section{G\"odel-L\"ob Provability Logic}\label{sec:logical-prelims}

\input{logical-prelims.tex}

\section{Labeled and Tree Sequent Systems}\label{sec:label-calc}

\input{section-2.tex}

\section{Linearizing Tree Sequents in Proofs}\label{sec:linearizing}

\input{section-3.tex}

\section{Sequent Systems and Correspondences}\label{sec:main-theorem}

\input{section-4.tex}

\section{Concluding Remarks}\label{sec:conclusion}

\input{conclusion.tex}

%%
%% Bibliography
%%

%% Please use bibtex, 

\bibliography{bibliography}

\appendix

\input{appendix.tex}

\end{document}

%% file: commands.tex
\usepackage{synttree, bussproofs, %stmaryrd, 
 rotating, xcolor, mathrsfs}
\usepackage{esvect, pgf, tikz, color}
\usetikzlibrary{arrows, automata}
\usepackage{algorithm2e}
\usepackage[all]{xy}
\usepackage{algorithmicx}
\usepackage{algpseudocode}% http://ctan.org/pkg/algorithmicx
\usepackage{scalerel,stackengine}
\usepackage{changepage, enumerate, proof, relsize,trfsigns,
comment}
\usepackage{multicol}
\usepackage{stmaryrd}

\newtheorem{obs}{Observation}

\makeatletter
\def\moverlay{\mathpalette\mov@rlay}
\def\mov@rlay#1#2{\leavevmode\vtop{%
   \baselineskip\z@skip \lineskiplimit-\maxdimen
   \ialign{\hfil$\m@th#1##$\hfil\cr#2\crcr}}}
\newcommand{\charfusion}[3][\mathord]{
    #1{\ifx#1\mathop\vphantom{#2}\fi
        \mathpalette\mov@rlay{#2\cr#3}
      }
    \ifx#1\mathop\expandafter\displaylimits\fi}
\makeatother

\newcommand{\block}{\mathrm{B}}
\newcommand{\blockii}{\mathrm{R}}

\newcommand{\permabove}{\text{\raise1pt\hbox{\rotatebox[origin=c]{45}{$\rightharpoonup$}}}}
\newcommand{\permbelow}{\text{\raise1pt\hbox{\rotatebox[origin=c]{45}{$\downarrow$}}}}
\newcommand{\gtgl}{\mathsf{G3KGL}}
\newcommand{\treegl}{\mathsf{CSGL}^{\!*}}
\newcommand{\lingl}{\mathsf{LNGL}}
\newcommand{\seqgl}{\mathsf{GL_{seq}}}
\newcommand{\seqglc}{\mathsf{GL_{circ}}}
\newcommand{\seqgli}{\mathsf{GL_{\infty}}}
\newcommand{\glbox}{\Box_{\mathsf{GL}}}

\newcommand{\logicgl}{\mathsf{GL}}
\newcommand{\semrel}{\models}
\newcommand{\rel}{\mathcal{R}}
\newcommand{\tel}{\mathcal{T}}

\newcommand{\idi}{\mathsf{id_{1}}}
\newcommand{\idii}{\mathsf{id_{2}}}
\newcommand{\disl}{\lor\mathsf{L}}

\newcommand{\disr}{{\lor}\mathsf{R}}
\newcommand{\boxl}{\Box\mathsf{L}}
\newcommand{\boxr}{\Box\mathsf{R}}
\newcommand{\boxlt}{\Box\mathsf{L}}
\newcommand{\boxrt}{\Box\mathsf{R}}
\newcommand{\negl}{\neg\mathsf{L}}
\newcommand{\negr}{\neg\mathsf{R}}

\newcommand{\irref}{\mathsf{ir}}
\newcommand{\trans}{\mathsf{tr}}
\newcommand{\fourru}{\mathsf{4L}}
\newcommand{\fourbox}{\Box_{\mathsf{4}}}
\newcommand{\lnsg}{\mathcal{G}}
\newcommand{\lnsh}{\mathcal{H}}

\newcommand{\sep}{\sslash}
\newcommand{\axi}{(\mathrm{A}1)}
\newcommand{\axii}{(\mathrm{A}2)}
\newcommand{\axiii}{(\mathrm{A}3)}
\newcommand{\axK}{(\mathrm{K})}
\newcommand{\axlob}{(\mathrm{L})}
\newcommand{\mpon}{(\mathrm{mp})}
\newcommand{\nec}{(\mathrm{n})}
\newcommand{\formint}{f}

\newcommand{\iffi}{\textit{iff} }
\newcommand{\tim}[1]{{\color{blue}T:  #1 }}

\renewcommand{\phi}{\varphi}

\newcommand{\wk}{\mathsf{w}}

\newcommand{\ctrl}{\mathsf{cL}}
\newcommand{\ctrr}{\mathsf{cR}}

\newcommand{\imp}{\rightarrow}

\newcommand{\len}[1]{||#1||}

\newcommand{\lsub}{(x/y)}
\newcommand{\prf}{\pi}

\newcommand{\ru}{\mathsf{r}}
\newcommand{\ruprime}{\mathsf{r}'}

\newcommand{\cut}{\mathsf{cut}}
\newcommand{\lobr}{\mathsf{L\ddot{o}b}}

\newtheorem{innercustomlem}{Lemma}
\newenvironment{customlem}[1]
  {\renewcommand\theinnercustomlem{#1}\innercustomlem}
  {\endinnercustomlem}
  
  \newtheorem{innercustomthm}{Theorem}
\newenvironment{customthm}[1]
  {\renewcommand\theinnercustomthm{#1}\innercustomthm}
  {\endinnercustomthm}

\SetKwFor{Loop}{Loop}{}{EndLoop}

\newcommand{\lab}{\mathrm{Lab}}

\newcommand{\propset}{\mathsf{Prop}}

\newcommand{\lang}{\mathcal{L}}

\newcommand{\id}{\mathsf{id}}

\newcommand{\sar}{\vdash}
\newcommand{\seq}{S}

\newcommand{\tseq}{T}

\newcommand{\sect}{Section} %Section
\newcommand{\thm}{Theorem} %Theorem
\newcommand{\fig}{Figure} %Figure
\newcommand{\ifandonlyif}{iff } %if and only if
\newcommand{\assign}{\mu}
\newcommand{\assu}{\mathrm{A}}
\newcommand{\bassu}{\mathrm{BA}}

%% file: intro.tex
Provability logics are a class of modal logics where the $\Box$ operator is read as `it is provable that' in some arithmetical theory. One of the most prominent provability logics is Gödel-Löb logic ($\logicgl$), which arose out of the work of Löb, who formulated a set of conditions on the provability predicate of Peano Arithmetic (PA). The logic $\logicgl$ can be axiomatized as an extension of the basic modal logic $\mathsf{K}$ by the single axiom $\Box (\Box \phi \rightarrow \phi) \rightarrow \Box \phi$, called \emph{Löb's axiom}. It is well-known that the axioms of $\logicgl$ are sound and complete relative to transitive and conversely-wellfounded relational models~\cite{Seg71}. In a landmark result, Solovay~\cite{Sol76} remarkably showed that $\logicgl$ is complete for PA's provability logic, i.e., $\logicgl$ proves everything that PA can prove about its own provability predicate. 

%non-first order definable frame condition (discuss as to why non standard sequent systems breaking certain properties e.g. Wansing prop's such as separation (loook these up))

%xAvr84

The logic $\logicgl$ enjoys a rich structural proof theory, possessing a number of cut-free sequent-style systems. %, each of which incorporates the bureaucracy inherent in proof construction in a distinct way. 
Sequent systems in the style of Gentzen were originally provided by Sambin and Valentini in the early 1980s~\cite{SamVal80,SamVal82}; see also Avron~\cite{Avr84}. (NB. In this work, we take a \emph{Gentzen system} to be a proof system whose rules operate over \emph{Gentzen sequents}, i.e., expressions of the form $\phi_{1}, \ldots, \phi_{n} \vdash \psi_{1}, \ldots, \psi_{k}$ such that $\phi_{i}$ and $\psi_{j}$ are logical formulae.) Since then, a handful of alternative systems have been introduced, each of which either generalizes the structure of sequents or generalizes the notion of proof. The labeled sequent system $\gtgl$ was provided by Negri~\cite{Neg05} and uses \emph{labeled sequents} in proofs, which are binary graphs whose nodes are Gentzen sequents. In a similar vein, the tree-hypersequent system $\treegl$ was provided by Poggiolesi~\cite{Pog09b} and uses \emph{tree-hypersequents} in proofs, which are trees whose nodes are Gentzen sequents. The use of \emph{(types of) graphs} of Gentzen sequents in proofs (as in $\gtgl$ and $\treegl$) allows for the systems to possess properties beyond those of the original Gentzen systems~\cite{SamVal80,SamVal82}. For example, both $\gtgl$ and $\treegl$ enjoy invertibility of all rules, rules are symmetric (i.e., for each logical connective, there is at least one rule that introduces it in the antecedent and at least one rule that introduces it in the consequent of the rule's conclusion),
%, which simplifies syntactic cut-elimination (cf.~\cite{Pog09b}), 
and the close connection between the syntax of such sequents and $\logicgl$'s relational semantics makes such systems suitable for counter-model extraction. 

%\footnote{For a discussion of symmetry and other desirable proof-theoretic criteria, see Wansing~\cite{Wan02}.} 

Rather than generalizing the structure of sequents, Shamkanov~\cite{Sha14} showed that one could obtain alternative cut-free sequent systems for $\logicgl$ by generalizing the structure of proofs. In particular, by taking the sequent calculus for the modal logic $\mathsf{K4}$ and allowing for \emph{non-wellfounded proofs}, one obtains a non-wellfounded sequent system $\seqgli$ for $\logicgl$. Non-wellfounded proofs were introduced to capture (co)inductive reasoning, and are potentially infinite trees of sequents such that (1) every parent node is the conclusion of a rule with its children the corresponding premises and (2) infinite branches satisfy a certain \emph{progress condition}, which ensures soundness (cf.~\cite{BroSim10,DasGieMar24,NiwWal96,Sha14}). For $\logicgl$, non-wellfounded proofs correspond to regular trees (i.e., only contain finitely many distinct sub-trees), which means such proofs can be `folded' into finite trees of sequents such that leaves are `linked' to internal nodes of the tree, giving rise to \emph{cyclic proofs} (cf.~\cite{AfsLei17,Bro05,Bro07,DasGir23}). Shamkanov~\cite{Sha14} additionally showed that one could obtain a cut-free cyclic sequent system $\seqglc$ for $\logicgl$ by allowing cyclic proofs in $\mathsf{K4}$'s sequent calculus, which was then used to provide the first syntactic proof of the 
Lyndon interpolation property 
for $\logicgl$.

Due to the diversity in $\logicgl$'s proof theory, it is natural to wonder about the relationships between the various systems that have been introduced. Typically, proof systems are related by means of \emph{proof transformations}, which are functions that map proofs from one calculus into another, are sensitive to the structure of the input proof, and operate syntactically by permuting rules, replacing rules, or adding/deleting sequent structure in the input proof to yield the output proof. Studying proof transformations between sequent systems is a beneficial enterprise as it lets one transfer results from one system to another, thus alleviating the need of independent proofs in each system (e.g.,~\cite{CiaLyoRamTiu21,GorRam12}). Moreover, one can measure the relative sizes or certain characteristics of proofs, giving insight into which systems are better suited for specific (automated) reasoning tasks, and letting one `toggle' between differing formalisms when one is better suited for a task than another (e.g.,~\cite{LyoTiuGorClo20,Lyo21thesis}). 

Indeed, the question of the relationship between the labeled and tree-hypersequent systems %$\gtgl$ and $\treegl$ 
was asked by Poggiolesi~\cite{Pog09b} and answered in full by Gor{\'e} and Ramanayake~\cite{GorRam12}, whose work provided constructive mappings of proofs between the systems. % by means of syntactic transformations. 
Similarly, Shamkanov~\cite{Sha14} provided syntactic mappings of proofs between the systems $\seqgli$ and $\seqglc$, and the Gentzen system $\seqgl$ (an equivalent reformulation of Sambin and Valentini's systems~\cite{SamVal80,SamVal82}). Nevertheless, the inter-translatability of proofs between the former two structural sequent systems ($\gtgl$ and $\treegl$) and the latter three sequent systems ($\seqgl$, $\seqgli$, and $\seqglc$) has yet to be identified, and presents a non-trivial open problem that we solve in this paper. Thus, our first contribution in this paper is to `complete the picture' and establish complete correspondences between the above five mentioned sequent systems by means of syntactic proof transformations. %We do this by 

There is an inherent difficulty in transforming proofs that use structural sequents (e.g., labeled sequents or tree-hypersequents) into proofs that use Gentzen sequents. This is due to the fact that structural sequents are (types of) graphs of Gentzen sequents, and thus, possess a more complicated structure that must be properly `shed' during proof transformations~\cite{Lyo21thesis,Tic23}. To overcome this difficulty and define proof transformations from $\gtgl$ and $\treegl$ to $\seqgl$, we rely on three techniques: first, we show how to restructure proofs in $\treegl$ so that they are \emph{end-active} (cf.~\cite{Lel15}), meaning rules only affect data at leaves or parents of leaves in tree-hypersequents. Second, we introduce a novel \emph{linearization technique}, whereby we show how to shed the tree structure of tree-hypersequents in end-active proofs, yielding a proof consisting solely of \emph{linear nested sequents}~\cite{Lel15}, i.e., lines whose nodes are Gentzen sequents. Linear nested sequents were introduced as an alternative (albeit equivalent) formalism to 2-sequents~\cite{Mas92,Mas93} that allows for sequent systems with complexity-optimal proof-search that also retain fundamental admissibility and invertibility properties~\cite{Lel15}. The presented linearization technique is new and shows how to extract linear nested sequent systems from tree-hypersequent systems, serving as the second contribution of this paper. We conjecture that this method can be generalized and applied in other settings to provide new linear nested sequent systems for modal and related logics. The technique also yields the first (cut-free) linear nested sequent calculus for $\logicgl$, which we dub $\lingl$, and which is the third contribution of this paper. Last, we show that proofs in $\lingl$ can be put into a specific normal form that repeats in stages of modal and local rules. Such proofs are translatable into $\seqgl$ proofs, which are translatable into $\gtgl$ proofs, thus establishing purely syntactic proof transformations between the most prominent sequent systems for $\logicgl$. %the five known systems $\gtgl$, $\treegl$, $\seqgl$, $\seqgli$, and $\seqglc$, as well as the (sixth) new linear nested sequent system $\lingl$. 
These proof transformations and systemic correspondences are summarized in Figure~\ref{fig:correspondence-diagram} below.
 %These correspondences are summarized below in Figure~\ref{fig:correspondence-diagram}.

\textbf{Outline of Paper.} In \cref{sec:logical-prelims}, we recall the language, semantics, and axioms of Gödel-Löb logic. In \cref{sec:label-calc}, we discuss Negri's labeled system $\gtgl$~\cite{Neg05}, Poggiolesi's tree-hypersequent system $\treegl$~\cite{Pog09b}, and Gor{\'e} and Ramanayake's correspondence result for the two systems~\cite{GorRam12}. In \cref{sec:linearizing}, we show how to put proofs in $\treegl$ into an end-active form (\cref{thm:proofs-to-end-active}) and specify our novel linearization method (\cref{thm:tree-to-lin}), which yields the new linear nested sequent system $\lingl$ for $\logicgl$. In \cref{sec:main-theorem}, we recall the sequent calculus $\seqgl$ due to Sambin and Valentini~\cite{SamVal80,SamVal82}, Shamkanov's non-wellfounded system $\seqgli$ and cyclic system $\seqglc$, as well as Shamkanov's correspondence result for the aforementioned three systems~\cite{Sha14}. We then show how to transform proofs in $\lingl$ into proofs in $\seqgl$ (\cref{thm:lin-to-seq}) and how to transform proofs in $\seqgl$ into proofs in $\gtgl$ (\cref{thm:seq-to-lab}). This establishes a six-way correspondence between the systems $\gtgl$, $\treegl$, $\seqgl$, $\seqgli$, $\seqglc$, and $\lingl$. % by means of syntactic transformations. %, thus unifying these disparate systems and works on the proof theory of $\logicgl$. 
 Last, in \cref{sec:conclusion}, we conclude and discuss future work. %We note that some proofs have been deferred to the appendix.
%We note that the appended version~\cite{Lyo24arxiv} contains additional proofs of results.

%appref

%***FIRST EVERY LNS EXTRACTION TECHNIQUE + LNS FOR GL + COMPLETES THE PICTURE OF RELATIONS BETWEEN STRUCTURAL, SEQ, and CYC SYSTEMS***
%method: (1) use end-active idea, (2) provide novel proof transformation that shows how to extract linear nested sequents from nested for GL. This provvides first ever (cut-free) LNS calculus for GL. Explain important properties of this calc: (i) whereas nested can be exponential in size of input formula, LNS is polynomial. 
%- from above, we go to sequents by relying on a special form of LNS proofs that organizes modal inferences into "block"
%- picture completed by transforming sequent calc proof into gtgl 
%we bring together disparate works on proof transformations in GL to "complete the picture" and fully specify how proofs in these various systems relate to one another

\begin{figure}[h]

\begin{minipage}{0.45\textwidth}
\begin{tabular}{c}
\xymatrix@C+=5em@R=4em{
\gtgl\ar@{<->}[r]^{\cite{GorRam12}}	
&  \treegl\ar@{->}[r]^{\textrm{Thm.}~\ref{thm:tree-to-lin}} 
& \lingl\ar@{->}[dl]_{\textrm{Thm.}~\ref{thm:lin-to-seq}} \\
\seqgli\ar@{<->}[r]^{\cite{Sha14}} 
& \seqgl\ar@{->}[ul]_{\textrm{Thm.}~\ref{thm:seq-to-lab}}
& \seqglc\ar@{<->}[l]_{\cite{Sha14}}
} %@/^2.0pc/
\end{tabular}
\end{minipage}
\begin{minipage}{0.05\textwidth}
\ \\
\end{minipage}
\begin{minipage}{0.45\textwidth}
\bgroup
\def\arraystretch{.85}
\begin{tabular}{|c|c|}
\hline
Name & Type of System\\
\hline
$\gtgl$ & Labeled Sequent System\\
$\treegl$ & Tree-Hypersequent System\\
$\lingl$ & Linear Nested Sequent System\\
$\seqgl$ & Gentzen System\\
$\seqgli$ & Non-Wellfounded Sequent System\\
$\seqglc$ & Cyclic Sequent System\\
\hline
\end{tabular}
\egroup
\end{minipage}

\caption{Proof transformations and correspondences between sequent systems for $\logicgl$.\label{fig:correspondence-diagram}}
\end{figure}

%% file: logical-prelims.tex
We let $\propset := \{p,q,r,\ldots\}$ be a countable set of \emph{propositional atoms} and define the language $\lang$ to be the set of all formulae generated via the following grammar in BNF:
$$
\phi ::= p \ | \ \neg \phi \ | \ \phi \lor \phi \ | \ \Box \phi
$$
where $p$ ranges over $\propset$. We use $\phi$, $\psi$, $\chi$, $\ldots$ to denote formulae in $\lang$ and %define $\Box^{n} \phi$ accordingly: $\Box^{0} \phi := \phi$ and $\Box^{n+1} \phi := \Box \Box^{n} \phi$. Furthermore, we 
%take
define $\phi \land \psi := \neg (\neg \phi \lor \neg \psi)$ and $\phi \imp \psi := \neg \phi \lor \psi$ as usual. % to be defined as usual.

%Furthermore, we define the negation $\negnnf{\phi}$ of a formula $\phi$ as follows: $\negnnf{p} = \overline{p}$, $\negnnf{\notp} = p$, $\negnnf{\phi \lor \psi} = \negnnf{\phi} \land \negnnf{\psi}$, $\negnnf{\phi \land \psi} = \negnnf{\phi} \lor \negnnf{\psi}$, $\negnnf{\pdldia{\pra} \phi} = \pdlbox{\pra} \negnnf{\phi}$, and $\negnnf{\pdlbox{\pra} \phi} = \pdldia{\pra} \negnnf{\phi}$. Finally, we let $\phi \imp \psi = \negnnf{\phi} \lor \psi$.

\begin{definition}[Model]\label{def:model} We define a \emph{model} to be a tuple $M = (W,R,V)$ such that
\begin{itemize}

\item $W$ is a non-empty set of worlds $w$, $u$, $v$, $\ldots$ (occasionally annotated);

\item $R \subseteq W \times W$ is transitive and conversely-wellfounded;\footnote{We note that $R$ is conversely-wellfounded \iffi it does not contain any infinite ascending $R$-chains.}

\item $V : \propset \to 2^{W}$ is a \emph{valuation function}.

\end{itemize}
\end{definition}

\begin{definition}[Semantic Clauses]\label{def:semantics} 
We define the \emph{satisfaction} of a formula $\phi$ in a model $M$ at world $w$, written $M,w \semrel \phi$, recursively as follows:
\begin{itemize}

\item $M,w \semrel p$ \ifandonlyif $w \in V(p)$;

\item $M,w \semrel \neg \phi$ \ifandonlyif $M,w \not\semrel \phi$;

\item $M,w \semrel \phi \lor \psi$ \ifandonlyif $M,w \semrel \phi$ or $M,w \semrel \psi$;

\item $M,w \semrel \Box \phi$ \ifandonlyif $\forall u \in W$, if $(w,u) \in R$, then $M,u \semrel \phi$;

\item $M \semrel  \phi$ \ifandonlyif $\forall w \in W$, $M,w \semrel \phi$.

\end{itemize}
We write $ \semrel \phi$ and say that $\phi$ is \emph{valid} \iffi for all models $M$, $M \semrel \phi$. %; we say that $\phi$ is \emph{invalid}, and write $\not\models \phi$, \iffi $\phi$ is not valid. We define 
\emph{G{\"o}del-L{\"o}b logic ($\logicgl$)} is defined to be the set $\logicgl \subset \lang$ of all valid formulae.
\end{definition}

As shown by Segerberg~\cite{Seg71}, the logic $\logicgl$ can be axiomatized by extending the axioms of the modal logic $\mathsf{K}$ with Löb's axiom $\Box(\Box \phi \imp \phi) \imp \Box \phi$.

\iffalse
\begin{align*}
\axi &  = \phi \rightarrow (\psi \rightarrow \phi)\\
\axii & = (\phi \rightarrow (\psi \rightarrow \chi)) \rightarrow ((\phi \rightarrow \psi) \rightarrow (\phi \rightarrow \chi))\\
\axiii & = (\neg \phi \rightarrow \neg \psi) \rightarrow (\psi \rightarrow \phi)\\
\axK & = \Box(\phi \imp \psi) \imp (\Box \phi \imp \Box \psi)\\
\axlob & = \Box(\Box \phi \imp \phi) \imp \Box \phi
\end{align*}

\begin{center}
\begin{tabular}{c c}
\AxiomC{$\phi$}\AxiomC{$\phi \rightarrow \psi$}\RightLabel{$\mpon$}\BinaryInfC{$\psi$}\DisplayProof

&

\AxiomC{$\phi$}\RightLabel{$\nec$}\UnaryInfC{$\Box \phi$}\DisplayProof
\end{tabular}
\end{center}
\fi

%% file: section-2.tex
In this section, we review the labeled sequent calculus $\gtgl$ by Negri~\cite{Neg05} and its correspondence (proven by Gor{\'e} and Ramanayake~\cite{GorRam12}) with a notational variant of Poggiolesi's tree-hypersequent system for $\logicgl$~\cite{Pog09b}. \emph{Labeled sequents} are binary graphs of traditional Gentzen sequents, which encode the relational semantics of a logic directly in the syntax of sequents. %meaning they are permitted to be disconnected or contain cycles. 
 The formalism of labeled sequents has been extensively studied with the inception of the formalism dating back to the work of Kanger~\cite{Kan57} and achieving its modern form in the work of Simpson~\cite{Sim94}. It has been shown that labeled sequent systems can capture sizable and diverse classes of logics in a cut-free manner while exhibiting fundamental properties such as the admissibility of various structural rules and the invertibility of rules~\cite{Sim94,Vig00,Neg05}.

By contrast, \emph{tree-hypersequents}, which are more traditionally known as \emph{nested sequents}, are trees of Gentzen sequents. %\footnote{We remark that nested sequents are known to be equivalent to prefixed tableaux~\cite{Fit14}.} 
The formalism was introduced independently by Kashima~\cite{Kas94} and Bull~\cite{Bul92}  with further influential works provided by Br{\"u}nnler~\cite{Bru09} and Poggiolesi~\cite{Pog09,Pog09b}. Such systems arose out of a call for cut-free sequent-style systems for logics not known to possess a cut-free Gentzen system, such as the tense logic $\mathsf{K_{t}}$ and the modal logic $\mathsf{S5}$. Like labeled sequents, tree-hypersequent systems exhibit fundamental admissibility and invertibility properties, having been defined for large classes of various logics such as tense logics~\cite{GorPosTiu08}, intuitionistic modal logics~\cite{Str13,Lyo21a}, and first-order non-classical logics~\cite{Fit14,LyoOrl23}.

As observed by Gor{\'e} and Ramanayake~\cite{GorRam12}, restricting labeled sequents to be trees, rather than more general, binary graphs (which may be disconnected or include cycles), yields labeled tree sequents (cf.~\cite{IshKik07}), which are a notational variant of tree-hypersequents/nested sequents. Via this observation, the authors established bi-directional proof transformations between %Negri's labeled sequent calculus and Poggiolesi's tree-hypersequent calculus 
labeled and tree-hypersequent proofs for $\logicgl$. Proof theoretic correspondences between labeled and nested systems for various other logics have been established in recent years as well; e.g., for tense logics~\cite{CiaLyoRamTiu21}, first-order intuitionistic logics~\cite{Lyo21thesis,Lyo21}, and intuitionistic modal logics~\cite{Lyo21a}. Recently, it was proven in a general setting that correspondences between labeled and nested systems are a product of two underlying proof transformation techniques, structural rule elimination and %structural rule 
introduction, and that (Horn) labeled and nested systems tend to come in pairs, being dual to one another~\cite{LyoOst24}.

Reducing Negri's labeled sequent system to one that uses trees, as opposed to binary graphs, is the first step in establishing syntactic correspondences between the various sequent systems for $\logicgl$. As shown in the sequel, we will systematically reduce the structure of sequents in proofs: first, going from binary graphs of Gentzen sequents to trees of Gentzen sequents (this section), then from trees of Gentzen sequents to lines of Gentzen sequents (\sect~\ref{sec:linearizing}), and last from lines of Gentzen sequents to Gentzen sequents themselves, which are easily embedded in labeled sequent proofs, completing the circuit of correspondences (\sect~\ref{sec:main-theorem}). This yields correspondences between the most widely regarded sequent systems for $\logicgl$, as depicted in \fig~\ref{fig:correspondence-diagram}.

\begin{figure}[t]
%\noindent\hrule

\begin{center}
\begin{tabular}{c c}
\AxiomC{}
\RightLabel{$\id$}
\UnaryInfC{$\rel, \Gamma, x : p \sar x : p, \Delta$}
\DisplayProof

&

\AxiomC{ }
\RightLabel{$\irref$}
\UnaryInfC{$\rel, xRx, \Gamma \sar \Delta$}
\DisplayProof
\end{tabular}
\end{center}

\begin{center}
\begin{tabular}{c c}
\AxiomC{$\rel, xRy, yRz, xRz, \Gamma \sar \Delta$}
\RightLabel{$\trans$}
\UnaryInfC{$\rel, xRy, yRz, \Gamma \sar \Delta$}
\DisplayProof

&

\AxiomC{$\rel, \Gamma, x : \phi \sar \Delta$}
\AxiomC{$\rel, \Gamma, x : \psi \sar \Delta$}
\RightLabel{$\disl$}
\BinaryInfC{$\rel, \Gamma, x : \phi \lor \psi \sar \Delta$}
\DisplayProof
\end{tabular}
\end{center}

\begin{center}
\begin{tabular}{c c c}
\AxiomC{$\rel, \Gamma \sar x : \phi, \Delta$}
\RightLabel{$\negl$}
\UnaryInfC{$\rel, \Gamma, x : \neg \phi \sar \Delta$}
\DisplayProof

&

\AxiomC{$\rel, \Gamma, x : \phi \sar \Delta$}
\RightLabel{$\negr$}
\UnaryInfC{$\rel, \Gamma \sar x : \neg \phi, \Delta$}
\DisplayProof

&

\AxiomC{$\rel, \Gamma \sar x : \phi, x : \psi, \Delta$}
\RightLabel{$\disr$}
\UnaryInfC{$\rel, \Gamma \sar x : \phi \lor \psi, \Delta$}
\DisplayProof
\end{tabular}
\end{center}

\begin{center}
\begin{tabular}{c c}
%\AxiomC{$\rel, xRy, \Gamma, x : \Box \phi \sar y : \Box \phi, \Delta$}
\AxiomC{$\rel, xRy, \Gamma, x : \Box \phi, y : \phi \sar \Delta$}
\RightLabel{$\boxl$}
\UnaryInfC{$\rel, xRy, \Gamma, x : \Box \phi \sar \Delta$}
\DisplayProof

&

\AxiomC{$\rel, xRy, \Gamma, y : \Box \phi \sar y : \phi, \Delta$}
\RightLabel{$\boxr^{\dag}$}
\UnaryInfC{$\rel, \Gamma \sar x : \Box \phi, \Delta$}
\DisplayProof
\end{tabular}
\end{center}

%\resizebox{\columnwidth}{!}{
%\noindent\rule{38em}{0.4pt}
\caption{Labeled Sequent Calculus $\gtgl$ for $\logicgl$. The $\boxr$ rule is subject to a side condition $\dag$, namely, the rule is applicable only if the label $y$ is fresh.\label{fig:lab-calc}}
\end{figure}

\subsection{Labeled Sequents}\label{subsec:lab-seq}

We let $\lab = \{x, y, z, \ldots\}$ be a countably infinite set of \emph{labels}, define a \emph{relational atom} to be an expression of the form $xRy$ with $x,y \in \lab$, and define a \emph{labeled formula} to be an expression of the form $x : \phi$ such that $x \in \lab$ and $\phi \in \lang$. We use upper-case Greek letters $\Gamma, \Delta, \Sigma, \ldots$ to denote finite multisets of labeled formulae. For a set $\rel$ of relational atoms and multiset $\Gamma$ of labeled formulae, we let $\lab(\rel)$, $\lab(\Gamma)$, and $\lab(\rel,\Gamma)$ be the sets of all labels occurring therein. For a multiset $\Gamma$ of labeled formulae, we define the multiset $\Gamma(x) := \{\phi \ | \ x : \phi \in \Gamma\}$, for a multiset of formulae $\Gamma := \phi_{1}, \ldots, \phi_{n}$, we define $x : \Gamma := x : \phi_{1}, \ldots, x : \phi_{n}$, and for multisets $\Gamma$ and $\Delta$ of labeled formulae, we let $\Gamma, \Delta$ denote the multiset union of the two. We define a \emph{labeled sequent} to be an expression of the form $\rel, \Gamma \sar \Delta$ with $\rel$ a set of relational atoms and $\Gamma, \Delta$ a multiset of labeled formulae. Given a labeled sequent $\rel, \Gamma \sar \Delta$, we refer to $\rel, \Gamma$ as the \emph{antecedent} and $\Delta$ as the \emph{consequent}. %We occasionally use $\lseq$ and annotated versions thereof to denote labeled sequents.
%Below, we define the interpretation of labeled sequents over models to make the reading of such sequents clear for the reader.
%Below, we define the interpretation of labeled sequents over models to clarify their reading.
%Below, we clarify the interpretation of labeled sequents by explaining how such sequents are evaluated over models.
Below, we clarify the interpretation of labeled sequents by explaining their evaluation over models.

\begin{definition}[Labeled Sequent Semantics] Let $M = (W,R,V)$ be a model. We define an \emph{$M$-assignment} to be a function $\assign \colon \lab \to W$. A labeled sequent $\rel, \Gamma \sar \Delta$ is \emph{satisfied} on $M$ with $M$-assignment $\assign$ \iffi if for all $xRy \in \rel$ and $x : \phi \in \Gamma$, $(\assign(x),\assign(y)) \in R$ and $M, \assign(x) \models \phi$, then there exists a $y : \psi \in \Delta$ such that $M, \assign(y) \models \psi$. A labeled sequent is defined to be \emph{valid} \iffi it is satisfied on all models $M$ with all $M$-assignments; a labeled sequent is defined to be \emph{invalid} otherwise. %\iffi it is not valid.
\end{definition}

\begin{figure}[t]
%\noindent\hrule

\begin{center}
\begin{tabular}{c c c}
\AxiomC{$\rel, \Gamma \sar \Delta$}
\RightLabel{$\lsub$}
\UnaryInfC{$\rel(x/y), \Gamma(x/y) \sar \Delta(x/y)$}
\DisplayProof

&

\AxiomC{$\rel, \Gamma \sar \Delta$}
\RightLabel{$\wk$}
\UnaryInfC{$\rel, \rel', \Gamma, \Gamma' \sar \Delta, \Delta'$}
\DisplayProof

&

\AxiomC{$\rel, \Gamma, x : \phi, x : \phi \sar \Delta$}
\RightLabel{$\ctrl$}
\UnaryInfC{$\rel, \Gamma, x : \phi \sar \Delta$}
\DisplayProof
\end{tabular}
\end{center}

\begin{center}
\begin{tabular}{c c}
\AxiomC{$\rel, \Gamma \sar x : \phi, x : \phi, \Delta$}
\RightLabel{$\ctrr$}
\UnaryInfC{$\rel, \Gamma \sar x : \phi, \Delta$}
\DisplayProof

&

\AxiomC{$\rel, \Gamma \sar x : \phi, \Delta$}
\AxiomC{$\rel, \Gamma, x : \phi \sar \Delta$}
\RightLabel{$\cut$}
\BinaryInfC{$\rel, \Gamma \sar \Delta$}
\DisplayProof
\end{tabular}
\end{center}

%\resizebox{\columnwidth}{!}{
%\noindent\rule{38em}{0.4pt}
\caption{Admissible rules.\label{fig:admiss-rules-lab}}
\end{figure}

Negri's labeled sequent calculus $\gtgl$ (adapted to our signature) is shown in \fig~\ref{fig:lab-calc}. The labeled calculus consists of two initial rules $\id$ and $\irref$. We refer to the conclusion of an initial rule as an \emph{initial sequent}. The $\trans$ rule is a structural rule that bottom-up adds transitive edges to labeled sequents, and the remaining rules form pairs of left and right logical rules, introducing complex logical formulae into either the antecedent or consequent of the rule's conclusion. We note that the $\boxr$ rule is subject to a side condition, namely, the label $y$ must be \emph{fresh} in any application of the rule, i.e., the label $y$ is forbidden to occur in the conclusion. We refer to the distinguished formulae in the conclusion (premises) of a rule as the \emph{principal formulae} (\emph{auxiliary formulae}, respectively). For example, $x : \Box \phi$ is principal in $\boxr$ and $xRy, x : \Box \phi, y : \phi$ are auxiliary.

\iffalse
\begin{remark}\label{rmk:box-left-premise} Negri's labeled system $\mathsf{G3GL}$ included the following $\boxl'$ rule rather than the $\boxl$ rule. However, the left premise of the $\boxl'$ rule is provable in $\gtgl$ using $\boxl$, $\boxr$, and $\trans$~\cite{Neg05}. We therefore opt to use the labeled calculus $\gtgl$ with the unary $\boxl$ rule rather than the binary $\boxl'$ rule to simplify our work.
\begin{center}
\AxiomC{$\rel, xRy, \Gamma, x : \Box \phi \sar y : \Box \phi, \Delta$}
\AxiomC{$\rel, xRy, \Gamma, x : \Box \phi, y : \phi \sar \Delta$}
\RightLabel{$\boxl'$}
\BinaryInfC{$\rel, xRy, \Gamma, x : \Box \phi \sar \Delta$}
\DisplayProof
\end{center}
\end{remark}
\fi

A \emph{derivation} of a labeled sequent $\rel, \Gamma \sar \Delta$ is defined to be a (potentially infinite) tree whose nodes are labeled with labeled sequents such that (1) $\rel, \Gamma \sar \Delta$ is the root of the tree and (2) each parent node is the conclusion of a rule with its children the corresponding premises. %, that is, derivations are read in a bottom-up manner. 
A \emph{proof} is a finite derivation such that every leaf is an instance of an initial sequent. We use $\prf$ (potentially annotated) to denote derivations and proofs throughout the remainder of the paper, and use this notation to denote derivations and proofs in other systems as well with the context determining the usage. The \emph{height} of a proof is defined as usual to be equal to the length of a maximal path from the root of the proof to an initial sequent.

\begin{remark}\label{rem:global-freshness}
We assume w.l.o.g. that every fresh variable used in a proof is \emph{globally fresh}, meaning there is a one-to-one correspondence between $\boxr$ applications and their fresh variables. This assumption is helpful, yet benign (cf.~\cite{Neg05}). 
\end{remark}

As shown by Negri~\cite{Neg05}, the various rules displayed in \fig~\ref{fig:admiss-rules-lab} are \emph{admissible} in $\gtgl$. Note that the $\lsub$ rule applies a \emph{label substitution} to the premise which replaces every occurrence of the label $y$ in a relational atom or labeled formula by $x$. We define a rule to be \emph{admissible} (\emph{height-preserving admissible}) \iffi if the premises of the rule have proofs (of height $h_{1}, \ldots, h_{n}$), then the conclusion of the rule has a proof (of height $h \leq \max\{h_{1}, \ldots, h_{n}\}$). We refer to a height-preserving admissible rule as \emph{hp-admissible}. Moreover, the non-initial rules of $\gtgl$ are \emph{height-preserving invertible}. If we let $\ru^{-1}_{i}$ be the $i$-inverse of the rule $\ru$ whose conclusion is the $i^{th}$ premise of the $n$-ary rule $\ru$ and premise is the conclusion of $\ru$, then we say that $\ru$ is \emph{(height-preserving) invertible} \iffi $\ru^{-1}_{i}$ is (height-preserving) admissible for each $1 \leq i \leq n$. We refer to height-preserving invertible rules as \emph{hp-invertible}.\label{def:i-inverse} The following theorem is due to Negri~\cite{Neg05}.

\begin{theorem}[$\gtgl$ Properties~\cite{Neg05}]\label{thm:lab-seq-properties} The labeled sequent calculus $\gtgl$ satisfies: % the following:
\begin{description}

\item[(1)] Each labeled sequent of the form $\rel, \Gamma, x : \phi \sar x : \phi, \Delta$ is provable in $\gtgl$;

\item[(2)] All non-initial rules are hp-invertible in $\gtgl$;

\item[(3)] The $\lsub$, $\wk$, $\ctrl$, and $\ctrr$ rules are hp-admissible in $\gtgl$;

\item[(4)] The $\cut$ rule is admissible in $\gtgl$;

\item[(5)] $\phi$ is valid \iffi $\sar x : \phi$ is provable in $\gtgl$.

\end{description}
\end{theorem}

\subsection{Labeled Tree Sequents}

A set $\tel$ of relational atoms is a \emph{tree} \iffi the graph $G(\tel) = (V,E)$ forms a tree, where $V = \{x \ | \ x \in \lab(\tel)\}$ and $E = \{(x,y) \ | \ xRy \in \tel\}$. %\footnote{A \emph{tree} is a graph such that there exists a unique directed path from a unique vertex $x$, called the \emph{root}, to every other vertex of the graph.}  
A \emph{tree sequent} is defined to be an expression of the form $\tel, \Gamma \sar \Delta$ such that (1) $\tel$ is a tree, (2) if $\tel \neq \emptyset$, then $\lab(\Gamma, \Delta) \subseteq \lab(\tel)$, and (3) if $\tel = \emptyset$, then $|\lab(\Gamma,\Delta)| = 1$, i.e. all labeled formulae in $\Gamma, \Delta$ share the same label. We note that conditions (1)--(3) ensure that each tree sequent forms a connected graph that is indeed of a tree shape. We use $\tseq$ and annotated versions thereof to denote tree sequents. We define a \emph{flat sequent} to be a tree sequent of the form $\Gamma \sar \Delta$, that is, a flat sequent is a sequent $\Gamma \sar \Delta$ without relational atoms and where every labeled formula in $\Gamma,\Delta$ shares the same label. The \emph{root} of a tree sequent $\tel, \Gamma \sar \Delta$ is the label $x$ such that there exists a unique directed path of relational atoms in $\tel$ from $x$ to every other label $y \in \lab(\tel, \Gamma,\Delta)$; if $\tel = \emptyset$, then the root is the single label $x$ shared by all formulae in $\Gamma,\Delta$.

%It is straightforward to observe that e
Every tree sequent encodes a tree whose vertices are flat sequents. In other words, each tree sequent $\tseq = \tel, xRy_{1}, \ldots, xRy_{n}, \Gamma \sar \Delta$ such that $x$ is the root and $y_{1}, \ldots, y_{n}$ are all children of $x$ can be graphically depicted as a tree $tr_{x}(\tseq)$ of the form shown below:
\begin{center}
%\resizebox{\columnwidth}{!}{
\begin{tabular}{c c c}
\xymatrix@C+=-3.5em@R=1.5em{
		&  \overset{x}{\boxed{x : \Gamma(x) \sar x : \Delta(x)}}\ar@{->}[dl]\ar@{->}[dr] &   		\\
 tr_{y_{1}}(\tel, \Gamma \sar \Delta) & \hdots  & tr_{y_{n}}(\tel, \Gamma \sar \Delta)
}
\end{tabular}
%}
\end{center}

\iffalse
A useful feature of tree sequents is that they admit interpretation as a formula, meaning we can `lift' the semantics from $\logicgl$ formulae to sequents. Before defining the formula interpretation of tree sequents, let us define a useful notion: if $\tel$ is a tree such that $x \in \lab(\tel)$, then we let $\tel_{x}$ be the \emph{sub-tree} of $\tel$ rooted at $x$; we note that if $x$ is a leaf, then $\tel_{x} = \emptyset$. Given a tree sequent $\tel, \Gamma \sar \Delta$ with root $x \in \lab(\tel, \Gamma, \Delta)$ we define its formula interpretation $\formint_{x}(\tel, \Gamma \sar \Delta)$ recursively on the structure of $\tel$ as follows: 
\begin{flushleft}
$\bullet$ $\formint_{x}(\Gamma \sar \Delta) := \displaystyle{\bigwedge \Gamma(x) \rightarrow \bigvee \Delta(x)}$\\
$\bullet$ $\formint_{x}(\tel, xRy_{1}, \ldots, xRy_{n}, \Gamma \sar \Delta) := \displaystyle{\bigwege \Gamma(x) \rightarrow \bigvee \Delta(x) \lor \!\! \bigvee_{1 \leq i \leq n} \Box \formint_{y}(\tel_{y_{i}}, \Gamma \sar \Delta) }$
\end{flushleft}
A tree sequent $\tel, \Gamma \sar \Delta$ with root $x$ is \emph{(in)valid} \iffi $\formint_{x}(\tel, \Gamma \sar \Delta)$ is (in)valid. %If a sequent $\tel \sar \Gamma$ is valid, then we denote this as $\models \tel \sar \Gamma$, and if it is invalid, then we denote this as $\not\models \tel \sar \Gamma$.
\fi

\begin{definition}[Tree Sequent Calculus $\treegl$] We define $\treegl := (\gtgl \setminus \{\irref,\trans\}) \cup \{\fourru\}$, where the $\fourru$ is shown below and the rules of the calculus only operate over tree sequents.
\begin{center}
\AxiomC{$\tel, xRy, \Gamma, x : \Box \phi, y : \Box \phi \sar \Delta$}
\RightLabel{$\fourru$}
\UnaryInfC{$\tel, xRy, \Gamma, x : \Box \phi \sar \Delta$}
\DisplayProof
\end{center}
\end{definition}

The system $\treegl$ is a notational variant of Poggiolesi's eponymous tree-hypersequent system~\cite{Pog09b}, and thus, we identify the two systems with one another. The main difference between the two systems is notational: the system defined above uses tree sequents, which are tree-hypersequents `dressed' as labeled sequents~\cite{GorRam12}. %We also note that Poggiolesi's original tree-hypersequent system $\treegl$ uses a notational variant of the binary rule $\boxl'$ discussed in Remark~\ref{rmk:box-left-premise}. However, as with the labeled system $\gtgl$, the left premise of this rule is provable in $\treegl$. We have therefore opted to use the unary $\boxl$ rule in $\treegl$ to simplify our work and note that this change is benign.

%We define a derivation, proof, the height of a proof, (hp-)admissibility, and (hp-)invertibility for $\treegl$ analogous to how such notions are defined for $\gtgl$ and will apply these terms and concepts in the expected way to other sequent systems as well. 

\begin{remark} Derivations, proofs, the height of a proof, (hp-)admissibility, the $i$-inverse of a rule, and (hp-)invertibility are defined for $\treegl$ analogously to how such notions are defined for $\gtgl$. We will apply these terms and concepts in the expected way to other sequent systems as well to avoid repeating similar definitions.
\end{remark}

The system $\treegl$ differs from $\gtgl$ in that $\treegl$ only allows for tree sequents in proofs, lacks the structural rules $\irref$ and $\trans$, and includes the $\fourru$ rule. We refer to $\fourru$ and $\boxlt$ as \emph{propagation rules} (cf.~\cite{Fit72,CasCerGasHer97}) since the rules bottom-up propagate data forward along relational atoms, and we refer to $\negl$, $\negr$, $\disl$, and $\disr$ as \emph{local rules} since they only affect formulae locally at a single label. For any rule, we call the label $x$ labeling the principal formula in the conclusion the \emph{principal label}, for the $\fourru$, $\boxl$, and $\boxr$ rules, we refer to the label $y$ labeling the auxiliary formula(e) in the premise(s) as the \emph{auxiliary label}, and for local rules, the \emph{auxiliary label} is taken to be the same as the principal label since they are identical. Note that we define a label $x$ in a tree sequent $\tel, \Gamma \sar \Delta$ to be a \emph{leaf} \iffi $x$ is a leaf in $\tel$, and we define a label $x$ to be a \emph{pre-leaf} \iffi for all $y \in \lab(\tel)$, if $xRy \in \tel$, then $y$ is a leaf.

%$xRy$ occurs in $\tel$ and $y$ is a leaf.

Since the tree sequent calculus $\treegl$ is isomorphic to Poggiolesi's tree-hypersequent system, the two systems share the same properties. %The following theorem is therefore a consequence of the work of Poggiolesi~\cite{Pog09b}. 
 We note that in the setting of tree sequents the $\lsub$ and $\wk$ rules are less general than for labeled sequents. In particular, such rules are assumed to preserve the `tree shape' of tree sequents when applied. Nevertheless, the restricted forms of these rules are still hp-admissible in $\treegl$.

\begin{theorem}[$\treegl$ Properties~\cite{Pog09b}]\label{thm:tree-seq-properties} The tree sequent calculus $\treegl$ satisfies the following:
\begin{description}

\item[(1)] Each tree sequent of the form $\tel, \Gamma, x : \phi \sar x : \phi, \Delta$ is provable in $\treegl$;

\item[(2)] All non-initial rules are hp-invertible in $\treegl$;

\item[(3)] The $\lsub$, $\wk$, $\ctrl$, and $\ctrr$ rules are hp-admissible in $\treegl$;

\item[(4)] The $\cut$ rule is admissible in $\treegl$;

\item[(5)] $\phi$ is valid \iffi $\sar x : \phi$ is provable in $\treegl$.

\end{description}
\end{theorem}

Proofs in $\gtgl$ and $\treegl$ are inter-translatable with one another. This correspondence follows from the work of Gor{\'e} and Ramanayake~\cite{GorRam12} and is based on a couple of observations. First, the $\irref$ rule does not occur in $\gtgl$ proofs where the end sequent is a tree sequent. It is not difficult to see why this is the case: if one takes a proof of a tree sequent, then bottom-up applications of rules from $\gtgl$ will not allow for directed cycles to enter a sequent in a proof. This follows from the fact that the conclusion of the proof is a tree sequent, which is free of directed cycles, and each rule of $\gtgl$ either preserves relational atoms bottom-up, adds a single relational atom from a label $x$ to a fresh label $y$ in the case of the $\boxrt$ rule, or adds an undirected cycle in the case of $\trans$. Since the conclusion of $\irref$ contains a directed cycle $xRx$, such a sequent will never occur in such a proof.

\begin{obs}[\cite{GorRam12}]\label{obs:irref-admiss}
The $\irref$ rule does not occur in any $\gtgl$ proof of a tree sequent.
\end{obs}

Second, Gor{\'e} and Ramanayake~\cite{GorRam12} show that instances of $\trans$ can be eliminated from labeled proofs not containing $\irref$ and replaced by instances of $\fourru$. This elimination procedure can be used to map $\gtgl$ proofs such that the end sequent is a tree sequent to $\treegl$ proofs. Conversely, %if we let arbitrary labeled sequents appear in $\treegl$ proofs, 
it can be shown that $\fourru$ can be eliminated from $\treegl$ proofs, yielding $\gtgl$ proofs as a consequence. % and replaced by instances of $\trans$. 
%Second, Gor{\'e} and Ramanayake show that the structural rule $\trans$ is admissible in $\treegl$ and the $\fourru$ is admissible in $\gtgl$~\cite{GorRam12}. %To map proofs of tree sequents from $\gtgl$ to proofs in $\treegl$, one eliminates instances of $\trans$ by permuting ins. For the converse direction, the $\fourru$ is 
These elimination results can be proven syntactically, showing that proof transformations exist between $\treegl$ and $\gtgl$ for proofs of tree sequents as summarized in the theorem below. We refer the reader to~\cite{GorRam12} for the details.

\begin{theorem}[\cite{GorRam12}]\label{thm:lab-tree-equiv}
A tree sequent $\tseq$ is provable in $\gtgl$ \iffi it is provable in $\treegl$.
\end{theorem}

%% file: section-3.tex
We now show how to extract a \emph{linear nested sequent calculus} from $\treegl$, dubbed $\lingl$ (see \fig~\ref{fig:lin-calc}). To the best of the author's knowledge, this is the first linear nested sequent calculus for G\"odel-L\"ob provability logic. Linear nested sequents were introduced by Lellmann~\cite{Lel15} and are a finite representation of Masini's 2-sequents~\cite{Mas92}. Such systems operate over \emph{lines} of Gentzen sequents and have been used to provide cut-free systems for intermediate and modal logics~\cite{KuzLel18,Lel15,Mas92,Mas93}. 

The extraction of linear nested sequent proofs from tree sequent proofs takes place in three phases. In the first phase, we show how to transform any $\treegl$ proof into an \emph{end-active} proof, i.e., a proof such that principal and auxiliary formulae only occur at (pre-)leaves in tree sequents (cf.~\cite{Lel15}). In the second phase, we define our novel linearization technique, where we identify specific paths in tree sequents and `prune' sub-trees, yielding a linear nested sequent proof as the result. This technique is an additional contribution of this paper, and we conjecture that this technique can be used in other settings to extract linear nested sequent systems from tree sequent/nested sequent systems. In the third phase, we show how to `reshuffle' a linear nested sequent proof so that the proof proceeds in repetitive stages of local rules, propagation rules, and $\boxr$ rules, which we refer to as a proof in \emph{normal form}. This transformation is motivated by one provided in~\cite{PimLelRam19} for so-called \emph{basic nested systems}, which transforms proofs in a similar manner to extract Gentzen sequent proofs. Our transformation is distinct however as it works within the context of linear nested sequents. The simpler data structure used in linear nested sequents and the `end-active' shape of the rules in $\lingl$ simplifies the process of reshuffling proofs into normal form.

\begin{figure}[t]
%\noindent\hrule

\begin{center}
\begin{tabular}{c c c}
\AxiomC{$\phantom{\lnsg}$}
\RightLabel{$\id$}
\UnaryInfC{$\lnsg \sep \Gamma, p \sar p, \Delta$}
\DisplayProof

&

\AxiomC{$\lnsg \sep \Gamma, \phi \sar \Delta$}
\AxiomC{$\lnsg \sep \Gamma, \psi \sar \Delta$}
\RightLabel{$\disl$}
\BinaryInfC{$\lnsg \sep \Gamma, \phi \lor \psi \sar \Delta$}
\DisplayProof

&

\AxiomC{$\lnsg \sep \Gamma \sar \phi, \psi, \Delta$}
\RightLabel{$\disr$}
\UnaryInfC{$\lnsg \sep \Gamma \sar \phi \lor \psi, \Delta$}
\DisplayProof
\end{tabular}
\end{center}

\begin{center}
\begin{tabular}{c c}
\AxiomC{$\lnsg \sep \Gamma \sar \phi, \Delta $}
\RightLabel{$\negl$}
\UnaryInfC{$\lnsg \sep \Gamma, \neg \phi \sar \Delta $}
\DisplayProof

&

%\end{tabular}
%\end{center}

%\begin{center}
%\begin{tabular}{c c}

\AxiomC{$\lnsg \sep \Gamma, \phi \sar \Delta$}
\RightLabel{$\negr$}
\UnaryInfC{$\lnsg \sep \Gamma \sar \neg \phi, \Delta$}
\DisplayProof
\end{tabular}
\end{center}

\begin{center}
\begin{tabular}{c c c}
\AxiomC{$\lnsg \sep \Gamma, \Box \phi \sar \Delta \sep \Sigma, \Box \phi \sar \Pi$}
\RightLabel{$\fourru$}
\UnaryInfC{$\lnsg \sep \Gamma, \Box \phi \sar \Delta \sep \Sigma \sar \Pi$}
\DisplayProof

&

\AxiomC{$\lnsg \sep \Gamma, \Box \phi \sar \Delta \sep \Sigma, \phi \sar \Pi$}
\RightLabel{$\boxl$}
\UnaryInfC{$\lnsg \sep \Gamma, \Box \phi \sar \Delta \sep \Sigma \sar \Pi$}
\DisplayProof

&

\AxiomC{$\lnsg \sep \Gamma \sar \Delta \sep \Box \phi \sar \phi$}
\RightLabel{$\boxr$}
\UnaryInfC{$\lnsg \sep \Gamma \sar \Box \phi, \Delta$}
\DisplayProof
\end{tabular}
\end{center}

%\resizebox{\columnwidth}{!}{
%\noindent\rule{38em}{0.4pt}
\caption{Linear Nested Sequent Calculus $\lingl$ for $\logicgl$.\label{fig:lin-calc}}
\end{figure}

%\subsection{Linear Nested Sequents}

A \emph{linear nested sequent} is an expression of the form $\lnsg := \Gamma_{1} \sar \Delta_{1} \sep \cdots \sep \Gamma_{n} \sar \Delta_{n}$ such that $\Gamma_{i}$ and $\Delta_{i}$ are multisets of formulae from $\lang$ for $1 \leq i \leq n$. We use $\lnsg$, $\lnsh$, $\ldots$ to denote linear nested sequents and note that such sequents admit a formula interpretation:
$$\formint(\Gamma \sar \Delta) := \bigwedge \Gamma \imp \bigvee \Delta
\qquad
\formint(\Gamma \sar \Delta \sep \lnsg) := \bigwedge \Gamma \imp (\bigvee \Delta \lor \Box \formint(\lnsg))
$$
We define a linear nested sequent $\lnsg$ to be (in)valid \iffi $\formint(\lnsg)$ is (in)valid. The linear nested sequent calculus $\lingl$ consists of the rules shown in \fig~\ref{fig:lin-calc}. We take the $\negl$, $\negr$, $\disl$, and $\disr$ rules to be \emph{local rules} and the $\fourru$ and $\boxlt$ rules to be \emph{propagation rules} in $\lingl$. 
%Similar to before, we define a proof to be a finite tree of linear nested sequents respecting the rules of $\lingl$ and where every leaf in the proof is an instance of $\idi$ or $\idii$. Standard notions such as the height of a proof, (hp-)admissibility of rules, and (hp-)invertibility of rules are defined analogously as for $\gtgl$. 
 For $1 \leq i \leq n$, we refer to $\Gamma_{i} \sar \Delta_{i}$ as the $i$-component (or, as a \emph{component} more generally) of the linear nested sequent $\lnsg = \Gamma_{1} \sar \Delta_{1} \sep \cdots \sep \Gamma_{n} \sar \Delta_{n}$, we refer to $\Gamma_{n} \sar \Delta_{n}$ as the \emph{end component}, and we define the \emph{length} of $\lnsg$ to be $\len{\lnsg} := n$, i.e., the length of a linear nested sequent is equal to the number of its components. Comparing $\lingl$ to $\treegl$, one can see that $\lingl$ is the calculus $\treegl$ restricted to lines of Gentzen sequents and where rules only operate in the last two components. Making use of the formula translation, it is a basic exercise to show that if the conclusion of any rule is invalid, then at least one premise is invalid, i.e., $\lingl$ is sound.

\begin{theorem}\label{thm:soundness-lns}
If $\lnsg$ is provable in $\lingl$, then $\lnsg$ is valid.
\end{theorem}

When transforming $\treegl$ proofs into $\lingl$ proofs later on, it will be helpful to use the weakening rule $\wk$ shown in the lemma below. Observe that any application of $\wk$ to $\id$ yields an initial sequent, and $\wk$ permutes above every other rule of $\lingl$. As an immediate consequence, we have that $\wk$ is hp-admissible in $\lingl$.

\begin{lemma}\label{lem:wk-admiss-lns} The following weakening rule $\wk$ is hp-admissible in $\lingl$.
\begin{center}
\AxiomC{$\lnsg \sep \Gamma \sar \Delta \sep \lnsh$}
\RightLabel{$\wk$}
\UnaryInfC{$\lnsg \sep \Gamma, \Sigma \sar \Pi, \Delta \sep \lnsh$}
\DisplayProof
\end{center}
\end{lemma}

%Furthermore, 
Furthermore, we have that the $\fourru$ and $\boxlt$ rules are hp-invertible in $\lingl$ since the premises of each rule may be obtained from the conclusion by $\wk$.

%In addition, all rules of $\lingl$ with the exception of $\boxr$ can be shown hp-invertible.

\begin{lemma}\label{lem:hp-invert-lingl}
The $\fourru$ and $\boxlt$ rules are hp-invertible in $\lingl$.
%All non-initial rules in $\lingl$ with the exception of $\boxr$ are hp-invertible.
\end{lemma}

\paragraph*{From Tree Sequents to Linear Nested Sequents}\label{subsec:tree-to-lin-seq}

To extract $\lingl$ from $\treegl$, we first establish a set of rule permutation results, i.e., we show that rules of a certain form in $\treegl$ can always be permuted below other rules of a specific form. We note that a rule $\ru$ \emph{permutes below} a rule $\ru'$ %, written $\ru \permbelow \ru'$, 
whenever an application of $\ru$ followed by $\ru'$ in a proof can be replaced by an application of $\ru'$ (potentially preceded by an application of an $i$-inverse of $\ru$) followed by an application of $\ru$ to derive the same conclusion. To make this definition more concrete, we show (1) the permutation of a unary rule $\ru$ below a binary rule $\ru'$ below top-left, (2) the permutation of a binary rule $\ru$ below a unary rule $\ru'$ below top-right, (3) the permutation of a unary rule $\ru$ below a unary rule $\ru'$ below bottom-left, and (4) the permutation of a binary rule $\ru$ below a binary rule $\ru'$ below bottom-right. In (1) and (4), we note that the case where $\ru$ is applied to the right premise of $\ru'$ is symmetric. Furthermore, recall that $\ru^{-1}_{1}$ and $\ru^{-1}_{2}$ are the $1$- and $2$-inverses of $\ru$, respectively. (NB. For the definition of the \emph{$i$-inverse} of a rule, see \cref{subsec:lab-seq}.) We use (annotated versions of) the symbol $\seq$ below to indicate not only tree sequents, but linear nested sequents since we consider permutations of rules in $\lingl$ later on as well. 

\begin{center}
\begin{tabular}{c c c c c c}
\AxiomC{$\seq_{0}$}
\RightLabel{$\ru$}
\UnaryInfC{$\seq_{2}$}

\AxiomC{$\seq_{1}$}
\RightLabel{$\ru'$}
\BinaryInfC{$\seq_{3}$}
\DisplayProof

&

$\leadsto$

&

\AxiomC{$\seq_{0}$}

\AxiomC{$\seq_{1}$}
\RightLabel{$\ru^{-1}$}
\UnaryInfC{$\seq'$}

\RightLabel{$\ru'$}
\BinaryInfC{$\seq$}
\RightLabel{$\ru$}
\UnaryInfC{$\seq_{3}$}
\DisplayProof

&

\AxiomC{$\seq_{0}$}

\AxiomC{$\seq_{1}$}
\RightLabel{$\ru$}
\BinaryInfC{$\seq_{2}$}

\RightLabel{$\ru'$}
\UnaryInfC{$\seq_{3}$}
\DisplayProof

&

$\leadsto$

&

\AxiomC{$\seq_{0}$}
\RightLabel{$\ru'$}
\UnaryInfC{$\seq$}

\AxiomC{$\seq_{1}$}
\RightLabel{$\ru'$}
\UnaryInfC{$\seq'$}

\RightLabel{$\ru$}
\BinaryInfC{$\seq_{3}$}
\DisplayProof
\end{tabular}
\end{center}
\begin{center}
\begin{tabular}{c c c c c c}
\AxiomC{$\seq_{0}$}
\RightLabel{$\ru$}
\UnaryInfC{$\seq_{1}$}
\RightLabel{$\ru'$}
\UnaryInfC{$\seq_{2}$}
\DisplayProof

&

$\leadsto$

&

\AxiomC{$\seq_{0}$}
\RightLabel{$\ru'$}
\UnaryInfC{$\seq$}
\RightLabel{$\ru$}
\UnaryInfC{$\seq_{2}$}
\DisplayProof

&

\AxiomC{$\seq_{0}$}
\AxiomC{$\seq_{1}$}
\RightLabel{$\ru$}
\BinaryInfC{$\seq_{2}$}

\AxiomC{$\seq_{3}$}
\RightLabel{$\ru'$}
\BinaryInfC{$\seq_{4}$}
\DisplayProof

&

$\leadsto$

&

\AxiomC{$\seq_{0}$}

\AxiomC{$\seq_{3}$}
\RightLabel{$\ru^{-1}_{1}$}
\UnaryInfC{$\seq$}
\RightLabel{$\ru'$}
\BinaryInfC{$\seq'$}

\AxiomC{$\seq_{1}$}
\AxiomC{$\seq_{3}$}
\RightLabel{$\ru^{-1}_{2}$}
\UnaryInfC{$\seq''$}

\RightLabel{$\ru'$}
\BinaryInfC{$\seq'''$}

\RightLabel{$\ru$}
\BinaryInfC{$\seq_{4}$}
\DisplayProof
\end{tabular}
\end{center}
The various admissible rule permutations we describe are based on the notion of \emph{end-activity}, which is a property of rule applications where principal and auxiliary formulae only occur at (pre-)leaves in sequents. End-activity was first discussed by Lellmann~\cite{Lel15} in the context of mapping Gentzen sequent proofs into linear nested sequent proofs for non-classical logics.

\begin{definition}[End Active] A $\treegl$ proof is \emph{end-active} \iffi the following hold:
%We define a $\treegl$ proof to be \emph{end-active} \iffi the following conditions hold:
\begin{description}

\item[(1)] The principal label in every instance of $\id$ is a leaf;

\item[(2)] The principal label of each local rule is a leaf;

\item[(3)] The principal and auxiliary labels of a propagation rule are a pre-leaf and leaf, respectively.

%\item[(4)] The principal label of a $\boxr$ rule is a pre-leaf or leaf in the conclusion.

\end{description}
A rule $\ru$ is \emph{end-active} \iffi it satisfies its respective condition above; otherwise, the rule is \emph{non-end-active}. We note that we always take the $\boxr$ rule to be end-active. %A local rule or $\boxrt$ rule is \emph{independent} from a rule $\ru'$ \iffi its principal label is distinct from the auxiliary label of $\ru'$. A propagation rule $\ru$ is \emph{independent} from a local or propagation rule $\ru'$ \iffi the principal label of $\ru$ is distinct from the auxiliary label of $\ru'$ and is \emph{independent} from a $\boxrt$ rule \iffi its auxiliary label is distinct from the auxiliary label of $\boxrt$. %Last, a $\boxrt$ rule is \emph{independent} from a 
\end{definition}

As stated in the following lemma, the end-activity of sequential rule applications determines a set of permutation relationships between the rules of $\treegl$. The lemma is straightforward to prove, though tedious due to the number of cases; its proof can be found in the appendix. 
%its proof can be found in the online, appended version of this paper~\cite{Lyo24arxiv}.

%%%applab appref

\begin{lemma}\label{lem:permutations} The following permutations hold in $\treegl$: 
\begin{description}

%\item[(1)] If $\ru$ is a non-end-active local rule and $\ruprime$ is non-initial and end-active, then $\ru \permbelow \ruprime$;

%\item[(2)] if $\ru$ is a non-end-active propagation rule and $\ruprime$ is non-initial and end-active, then $\ru \permbelow \ruprime$;

%\item[(3)] if a $\boxrt$ rule is non-end-active and $\ru$ is non-initial and end-active, then $\boxrt \permbelow \ru$.

\item[(1)] If $\ru$ is a non-end-active local rule and $\ruprime$ is non-initial and end-active, then $\ru$ permutes below $\ruprime$ and $\ruprime$ remains end-active after this permutation;

\item[(2)] if $\ru$ is a non-end-active propagation rule and $\ruprime$ is non-initial and end-active, then $\ru$ permutes below $\ruprime$ and $\ruprime$ remains end-active after this permutation.
    
\end{description}
\end{lemma}

%Let us define a \emph{final-active} proof in $\treegl$ to be a proof such that the last inference is end-active. 
Using the above lemma, every proof $\prf$ in $\treegl$ of a sequent of the form $\sar x : \phi$ can be transformed into an end-active proof as follows: first, observe that the last inference in $\prf$ must be end-active since the proof ends with $\sar x : \phi$. By successively considering bottom-most instances of non-end-active local and propagation rules $\ru$ in $\prf$, we can repeatedly apply \cref{lem:permutations} to permute $\ru$ lower in the proof because all rules below $\ru$ are guaranteed to be end-active. By inspecting the rules of $\treegl$, we know that the trees within the tree sequents in $\prf$ will never `grow' but will `shrink' as they get closer to the conclusion of the proof, meaning, each non-end-active rule $\ru$ will eventually become end-active through successive downward permutations.\footnote{Observe that $\negl$, $\negr$, $\disl$, $\disr$, $\boxlt$, and $\fourru$ only affect the formulae associated with the label of a tree sequent, whereas $\boxrt$ top-down removes a relational atom from a tree sequent. Therefore, the number of relational atoms in tree sequents will decrease as we move down paths in $\treegl$ proofs from initial sequents to the conclusion.} %Moreover, permuting a non-end-active rule $\ru'$ below an end-active rule $\ru$ preserves the end-activity of the rule $\ru$.
This process will eventually terminate and yield a proof where all non-initial rules are end-active for the following two reasons: (1) As stated in the lemma above, permuting a non-end-active local or propagation rule $\ru'$ below an end-active rule $\ru$ preserves the end-activity of the rule $\ru$. (2) Although downward permutations may require the $i$-inverse of a rule to be applied above the permuted inferences, the hp-invertibility of all non-initial rules in $\treegl$ (see \thm~\ref{thm:tree-seq-properties}) ensures that the height of the proof does not grow after a downward rule permutation.

After all such downward permutations have been performed, the resulting proof is \emph{almost} end-active with the exception that initial rules may not be end-active. For example, as shown below left, it may be the case that $\id$ is non-end-active and followed by a rule $\ru$. Since $\ru$ is guaranteed to be end-active at this stage, we know that the auxiliary label of $\ru$ is distinct from $y$, meaning, the conclusion will be an instance of $\id$ as shown below right.
\begin{center}
\begin{tabular}{c c}
\AxiomC{}
\RightLabel{$\id$}
\UnaryInfC{$\tel, \Gamma, y : p \sar y : p, \Delta$}
\AxiomC{($\tel', \Gamma', y : p \sar y : p, \Delta'$)}
\RightLabel{$\ru$}
\BinaryInfC{$\tel'', \Gamma'', y : p  \sar y : p , \Delta''$}
\DisplayProof

&

\AxiomC{}
\RightLabel{$\id$}
\UnaryInfC{$\tel'', \Gamma'', y : p  \sar y : p , \Delta''$}
\DisplayProof
\end{tabular}
\end{center}
By replacing such rule applications $\ru$ by $\id$ instances, effectively `pushing' initial rules down in the proof, we will eventually obtain initial rules such that the label $y$ is auxiliary in the subsequent rule application, which will then be a leaf since all non-initial rules of the proof are end-active. Thus, every proof in $\treegl$ of a sequent of the form $\sar x : \phi$ can be transformed into an end-active proof. %every final-active proof in $\treegl$ can be transformed into an end-active proof. 

%\begin{remark}\label{rmk:prove-formula-is-final-active} As a corollary, we note that every proof $\prf$ in $\treegl$ of a sequent $\sar x : \phi$ can be transformed into an end-active proof as well. This is due to the fact that the last inference in $\prf$ must be end-active since the proof ends with $\sar x : \phi$, i.e., $\prf$ is final-active in this case.
%\end{remark}

\begin{theorem}\label{thm:proofs-to-end-active}
Each proof in $\treegl$ of a sequent of the form $\sar x : \phi$ can be transformed into an end-active proof.
\end{theorem}

\begin{theorem}\label{thm:tree-to-lin}
Each end-active proof in $\treegl$ can be transformed into a proof in $\lingl$.
%If $\prf$ is an end-active proof of a sequent $\sar x : \phi$ in $\treegl$, then $\prf$ can be transformed into a proof in $\lingl$.
\end{theorem}

\begin{proof} We prove that if there exists an end-active proof in $\treegl$ of a tree sequent $\tel, \Gamma \sar \Delta$, then there exists a path $x_{1}, \ldots, x_{n}$ of labels from the root $x_{1}$ to a leaf $x_{n}$ in $\tel, \Gamma \sar \Delta$ such that $\Gamma(x_{1}) \sar \Delta(x_{1}) \sep \cdots \sep \Gamma(x_{n}) \sar \Delta(x_{n})$ is provable in $\lingl$. We argue this by induction on the height of the end-active proof $\prf$ in $\treegl$.

\medskip

\noindent
\textit{Base case.} Suppose $\prf$ consists of a single application of $\id$, as shown below:
\begin{center}
\AxiomC{}
\RightLabel{$\id$}
\UnaryInfC{$\tel, \Gamma, x : p \sar x : p, \Delta$}
\DisplayProof
\end{center}
We know that each initial sequent is end-active, i.e., the label $x$ is a leaf. Therefore, since the sequent above is a tree sequent, there exists a path $y_{1}, \ldots, y_{n} = x$ of labels from the root $y_{1}$ to the leaf $x$. By using this path, we obtain an instance of $\id$ as shown below:
\begin{center}
\AxiomC{}
\RightLabel{$\id$}
\UnaryInfC{$\Gamma(y_{1}) \sar \Delta(y_{1}) \sep \cdots \sep \Gamma(x), p \sar p, \Delta(x)$}
\DisplayProof
\end{center}

\medskip

%%%applab
\noindent
\textit{Inductive step.} For the inductive hypothesis (IH), we assume that the claim holds for every end-active proof in $\treegl$ of height $h' \leq h$, and aim to show that the claim holds for proofs of height $h+1$. We let $\prf$ be of height $h+1$ and argue the cases where $\prf$ ends with $\boxlt$ or $\boxrt$ as the remaining cases are shown similarly. Some additional cases are given in the appendix. 
%Some additional cases can be found in the appended version of this paper~\cite{Lyo24arxiv}.

%%%appref

%For the inductive hypothesis (IH), we assume that if an end-active proof in $\treegl$ is of height $h' \leq h$, then it can be transformed into a proof in $\lingl$. We aim to show that if an end-active proof $\prf$ in $\treegl$ is of height $h+1$, then $\prf$ can be transformed into a proof in $\lingl$. We prove the claim when $\prf$ ends with a $\boxlt$ and $\boxrt$ inference as the remaining cases are shown similarly. Some additional cases are given in the appendix.

\medskip

\noindent
$\boxlt$. Let us suppose that $\prf$ ends with an instance of $\boxlt$ as shown below.
\begin{center}
%\AxiomC{$\tel, xRy, \Gamma, x : \Box \phi \sar y : \Box \phi, \Delta$}
\AxiomC{$\tel, xRy, \Gamma, x : \Box \phi, y : \phi \sar \Delta$}
\RightLabel{$\boxlt$}
\UnaryInfC{$\tel, xRy, \Gamma, x : \Box \phi \sar \Delta$}
\DisplayProof
\end{center}
By IH, we know there exists a path $y_{1}, \ldots, y_{n}$ of labels from the root $y_{1}$ to the leaf $y_{n}$ in the premise of $\boxlt$ such that $\lnsg = \Gamma(y_{1}) \sar \Delta(y_{1}) \sep \cdots \sep \Gamma(y_{n}) \sar \Delta(y_{n})$ is provable in $\lingl$. Since $\boxlt$ is end-active, we have three cases to consider: (1) neither $x$ nor $y$ occur along the path in the premise, (2) only $x$ occurs along the path in the premise, or (3) both $x$ and $y$ occur along the path in the premise. In cases (1) and (2), we translate the entire $\boxlt$ inference as the linear nested sequent $\lnsg$. In case (3), $\lnsg$ has the form of the premise shown below with $\Gamma(y_{n-1}) = \Sigma_{1}, \Box \phi$ and $\Gamma(y_{n}) = \Sigma_{2}, \phi$. A single application of $\boxl$ gives the desired result.
\begin{center}
\AxiomC{$\Gamma(y_{1}) \sar \Delta(y_{1}) \sep \cdots \sep \Sigma_{1}, \Box \phi \sar \Delta(y_{n-1}) \sep \Sigma_{2}, \phi \sar \Delta(y_{n})$}
\RightLabel{$\boxl$}
\UnaryInfC{$\Gamma(y_{1}) \sar \Delta(y_{1}) \sep \cdots \sep \Sigma_{1}, \Box \phi \sar \Delta(y_{n-1}) \sep \Sigma_{2} \sar \Delta(y_{n})$}
\DisplayProof
\end{center}
%In each case, the conclusion of $\boxl$ in $\lingl$ is obtained by taking either the linear nested sequent corresponding to the relevant path $y_{1}, \ldots, y_{n}$ or $z_{1}, \ldots, z_{k}$ in the conclusion of the $\boxlt$ instance above.

\medskip

\noindent
$\boxrt$. Let us suppose that $\prf$ ends with an instance of $\boxrt$ as shown below.
\begin{center}
\AxiomC{$\tel, xRy, \Gamma, y : \Box \phi \sar y : \phi, \Delta$}
\RightLabel{$\boxrt$}
\UnaryInfC{$\tel, \Gamma \sar x : \Box \phi, \Delta$}
\DisplayProof
\end{center}
By IH, we know there exists a path $y_{1}, \ldots, y_{n}$ of labels from the root $y_{1}$ to the leaf $y_{n}$ in the premise of $\boxrt$ such that $\Gamma(y_{1}) \sar \Delta(y_{1}) \sep \cdots \sep \Gamma(y_{n}) \sar \Delta(y_{n})$ is provable in $\lingl$. We have three cases to consider: either (1) neither $x$ nor $y$ occur along the path, (2) only $x$ occurs along the path, or (3) both $x$ and $y$ occur along the path. In case (1), we translate the entire $\boxrt$ instance as the single linear nested sequent $\lnsg$. In case (2), we know that $x = y_{i}$ for some $1 \leq i \leq n$. To obtain the desired conclusion, we apply the hp-admissible $\wk$ rule as shown below. Observe that the conclusion of the $\wk$ application below corresponds to the linear nested sequent obtained from the path $y_{1}, \ldots, y_{n}$ in the conclusion of the $\boxrt$ instance above.
\begin{center}
\AxiomC{$\Gamma(y_{1}) \sar \Delta(y_{1}) \sep \cdots \sep \Gamma(y_{i}) \sar \Delta(y_{i}) \sep \cdots \sep \Gamma(y_{n}) \sar \Delta(y_{n})$}
\RightLabel{$\wk$}
\UnaryInfC{$\Gamma(y_{1}) \sar \Delta(y_{1}) \sep \cdots \sep \Gamma(y_{i}) \sar \Box \phi, \Delta(y_{i}) \sep \cdots \sep \Gamma(y_{n}) \sar \Delta(y_{n})$}
\DisplayProof
\end{center}
Last, in case (3), we know that $x = y_{n-1}$ and $y = y_{n}$ due to the freshness condition imposed on the $\boxrt$ rule. In this case, $\lnsg$ has the form of the premise shown below, meaning, a single application of the $\boxr$ rule gives the linear nested sequent corresponding to the path $y_{1}, \ldots, y_{n}$ in the conclusion of the $\boxrt$ instance above.
\begin{center}
\AxiomC{$\Gamma(y_{1}) \sar \Delta(y_{1}) \sep \cdots \sep \Gamma(y_{n-1}) \sar \Delta(y_{n-1}) \sep \Box \phi \sar \phi$}
\RightLabel{$\boxr$}
\UnaryInfC{$\Gamma(y_{1}) \sar \Delta(y_{1}) \sep \cdots \sep \Gamma(y_{n-1}) \sar \Box \phi, \Delta(y_{n-1})$}
\DisplayProof
\end{center}
\end{proof}

The following is an immediate consequence of Theorems~\ref{thm:soundness-lns}, \ref{thm:proofs-to-end-active}, and \ref{thm:tree-to-lin}. % and Remark~\ref{rmk:prove-formula-is-final-active}.

\begin{corollary}[$\lingl$ Soundness and Completeness]
$\phi$ is valid \iffi $\sar \phi$ is provable in $\lingl$.
\end{corollary}

Last, we show that every $\lingl$ proof can be put into a \emph{normal form} (see \cref{def:normal-form} and \cref{thm:lingl-normalizable} below) such that (reading the proof bottom-up) $\boxrt$ instances are preceded by $\fourru$ instances, which are preceded by $\boxlt$ instances, which are preceded by local rule instances (or, initial rules). We will utilize this normal form in the next section to show that every $\lingl$ proof can be transformed into a Gentzen sequent proof (\cref{thm:lin-to-seq}). We let $\block$ be a set of $\lingl$ rules and define a \emph{block} to be a derivation that only uses rules from $\block$. We use the following notation to denote blocks, showing that the set $\block$ of rules derives $\lnsg$ from $\lnsg_{1}, \ldots, \lnsg_{n}$, and refer to $\lnsg_{1}, \ldots, \lnsg_{n}$ as the \emph{premises} of the block $\block$.
\begin{center}
\AxiomC{$\lnsg_{1}, \ldots, \lnsg_{n}$}
\RightLabel{$\block$}
\doubleLine
\UnaryInfC{$\lnsg$}
\DisplayProof
\end{center}

%(i) We can permute each $\boxrt$ rule down until the principal label $x$ in the conclusion is either a leaf or pre-leaf by \lem~\ref{lem:permutations} - (3). (ii) We can then permute $\fourru$ instances with $x$ principal and $y$ auxiliary down in $\prf$ until they occur above $\boxrt$ by Lemma~\ref{lem:permutations} - (2), (4), and (5). (iii) We can permute $\boxlt$ instances with $x$ principal and $y$ auxiliary down above the $\fourru$ instances by Lemma~\ref{lem:permutations} - (2), (4), and (6), and finally, (iv) we can permute local rules with $y$ principal down above the $\boxlt$ instances by Lemma~\ref{lem:permutations} - (1) and (6). Therefore, each $\treegl$ proof can be put into a \emph{normal form} (defined below) such that every rule instance is end-active and (reading the proof bottom-up) $\boxrt$ instances are preceded by $\fourru$ instances, which are preceded by $\boxlt$ instances, which are preceded by local rule instances (or, initial rules). %We note that a $\treegl$ proof in this form bottom-up processes tree sequents in a breadth-first manner.
%We define a \emph{block} of rules $\block$ to be a set of rule instances and use the following notation to denote that the set $\block$ of rule instances derives $\tseq$ from $\tseq_{1}, \ldots, \tseq_{n}$.
%\begin{center}
%\AxiomC{$\tseq_{1}, \ldots, \tseq_{n}$}
%\RightLabel{$\block$}
%\doubleLine
%\UnaryInfC{$\tseq$}
%\DisplayProof
%\end{center}
%We refer to $\tseq_{1}, \ldots, \tseq_{n}$ as the \emph{premises} of the block $\block$.

\begin{definition}[Normal Form]\label{def:normal-form} A proof in $\lingl$ is in \emph{normal form} \iffi each bottom-up $\boxrt$ application is derived from a block $\block$ of $\fourru$ rules, whose premise is derived from a block $\block'$ of $\boxlt$ rules, whose premise is derived from a block $\block''$ of local rules, as indicated below. %, where $\prf_{i}$ is shown below left.
\begin{center}
%\begin{tabular}{c c}
%\AxiomC{$\lnsg \sep \Gamma \sar \Delta \sep \Sigma_{1} \sar \Pi_{1}$, $\ldots$, $\lnsg \sep \Gamma \sar \Delta \sep \Sigma_{n} \sar \Pi_{n}$}
\AxiomC{$\lnsg \sep \Gamma \sar \Delta \sep \Sigma_{1} \sar \Pi_{1}$}
\AxiomC{$\ldots$}
\AxiomC{$\lnsg \sep \Gamma \sar \Delta \sep \Sigma_{n} \sar \Pi_{n}$}
\RightLabel{$\block''$}
\doubleLine
\TrinaryInfC{$\lnsg \sep \Gamma \sar \Delta \sep \Gamma', \Gamma'', \Box \phi \sar \phi$}
\RightLabel{$\block'$}
\doubleLine
\UnaryInfC{$\lnsg \sep \Gamma \sar \Delta \sep \Gamma', \Box \phi \sar \phi$}
\RightLabel{$\block$}
\doubleLine
\UnaryInfC{$\lnsg \sep \Gamma \sar \Delta \sep \Box \phi \sar \phi$}
\RightLabel{$\boxrt$}
\UnaryInfC{$\lnsg \sep \Gamma \sar \Box \phi, \Delta$}
\DisplayProof
%\end{tabular}
\end{center}
We refer to a block of rules of the above form as a \emph{complete block}, and refer to the portion of a complete block consisting of only $\boxr$, $\block$, and $\block'$ as a \emph{modal block}. %Note that the blocks $\block$, $\block'$, and $\block_{i}$ may be empty.
\end{definition}

As proven in the next section (\cref{thm:lin-to-seq}), every normal form proof in $\lingl$ can be transformed into a proof in Sambin and Valentini's Gentzen calculus $\seqgl$. Therefore, we need to show that every proof in $\lingl$ can be put into normal form. We prove this by making an observation about the structure of proofs in $\lingl$. Observe that local and propagation rules in $\lingl$ only affect the end component of linear nested sequents and preserve the length of such sequents, whereas the $\boxrt$ rule increases the length of a linear nested sequent by 1 when applied bottom-up. This implies that any $\lingl$ proof $\prf$ bottom-up proceeds in repetitive stages, as we now describe. Let $\prf$ be a proof in $\lingl$ with conclusion $\lnsg$ such that $\len{\lnsg} = n$. The conclusion $\lnsg$ is derived with a block $\block$ of local and propagation rules that only affect the $n$-component in inferences with the premises of the block $\block$ being initial rules or derived by applications of $\boxrt$ rules. These applications of $\boxrt$ rules will have premises of length $n+1$ and will be preceded by blocks $\block_{i}$ of local and propagation rules that only affect the $(n + 1)$-component in inferences. The premises of the $\block_{i}$ blocks will then either be initial rules or derived by applications of $\boxrt$ rules that have premises of length $n+2$, which are preceded by blocks of local and propagation rules that only affect the $(n + 2)$-component in inferences, and so on. Every proof in $\lingl$ will have this repetitive structure.

Let $\boxrt$ be applied in an $\lingl$ proof $\prf$ with premise $\lnsg \sep \Gamma \sar \Delta \sep \Box \phi \sar \phi$ of length $n$. We say that an instance of a local or propagation rule $\ru$ in $\prf$ is \emph{length-consistent} with $\boxrt$ \iffi the length of the conclusion of $\ru$ is equal to $n$. Based on the discussion above, we can see that for any $\boxrt$ application in a proof $\prf$, all length-consistent local and propagation rules will occur in a block $\block$ above the $\boxrt$ application with $\block$ free of other $\boxrt$ rules. It is not difficult to show that $\block$ can be transformed into a \emph{complete block} by (1) successively permuting $\fourru$ rules down into a block above $\boxrt$, and (2) successively permuting $\boxlt$ rules down above the $\fourru$ block. After the permutations from (1) and (2) have been carried out, the premise of the $\boxlt$ block will be derived by length-consistent local rule applications, showing that $\boxrt$ is preceded by a complete block. As these permutations can be performed for every $\boxrt$ rule in a proof, every proof can be put into normal form.

\begin{theorem}\label{thm:lingl-normalizable}
Every proof in $\lingl$ can be transformed into a proof in normal form.
\end{theorem}

%% file: section-4.tex
\begin{figure}[t]
%\noindent\hrule

\begin{center}
\begin{tabular}{c c c c}
\AxiomC{$\phantom{\Gamma}$}
\RightLabel{$\id$}
\UnaryInfC{$\Gamma, \phi \sar \phi, \Delta$}
\DisplayProof

&

\AxiomC{$\Gamma \sar \phi, \Delta $}
\RightLabel{$\negl$}
\UnaryInfC{$\Gamma, \neg \phi \sar \Delta $}
\DisplayProof

&

\AxiomC{$\Gamma, \phi \sar \Delta$}
\RightLabel{$\negr$}
\UnaryInfC{$\Gamma \sar \neg \phi, \Delta$}
\DisplayProof

&

\AxiomC{$\Gamma, \phi \sar \Delta$}
\AxiomC{$\Gamma, \psi \sar \Delta$}
\RightLabel{$\disl$}
\BinaryInfC{$\Gamma, \phi \lor \psi \sar \Delta$}
\DisplayProof
\end{tabular}
\end{center}

\begin{center}
\begin{tabular}{c c c}
\AxiomC{$\Gamma \sar \phi, \psi, \Delta$}
\RightLabel{$\disr$}
\UnaryInfC{$\Gamma \sar \phi \lor \psi, \Delta$}
\DisplayProof

&

\AxiomC{$\Box \Gamma, \Gamma, \Box \phi \sar \phi$}
\RightLabel{$\glbox$}
\UnaryInfC{$\Sigma, \Box \Gamma \sar \Box \phi, \Delta$}
\DisplayProof

&

\AxiomC{$\Box \Gamma, \Gamma \sar \phi$}
\RightLabel{$\fourbox$}
\UnaryInfC{$\Sigma, \Box \Gamma \sar \Box \phi, \Delta$}
\DisplayProof
\end{tabular}
\end{center}

%\resizebox{\columnwidth}{!}{
%\noindent\rule{38em}{0.4pt}
\caption{Sequent calculus rules.\label{fig:seq-calc}}
%\caption{Sequent Calculus $\seqgl$ for $\logicgl$.\label{fig:seq-calc}}
%\caption{The sequent calculus $\seqgl = \{\id, \negl, \negr, \disl, \disr, \glbox\}$ and $\mathsf{K4_{seq}} = \{\id, \negl, \negr, \disl, \disr, \fourbox\}$.\label{fig:seq-calc}}
\end{figure}

\subsection{Gentzen, Cyclic, and Non-Wellfounded Systems}

We use $\Gamma$, $\Delta$, $\Sigma$, $\ldots$ to denote finite multisets of formulae within the context of sequent systems. For a multiset $\Gamma := \phi_{1}, \ldots, \phi_{n}$, we define $\Box \Gamma := \Box \phi_{1}, \ldots, \Box \phi_{n}$. A \emph{sequent} is defined to be an expression of the form $\Gamma \sar \Delta$. %the \emph{formula interpretation} of a sequent $\Gamma \sar \Delta$ is defined in the usual way as $\formint(\Gamma \sar \Delta) := \bigwedge \Gamma \rightarrow \bigvee \Delta$. 
The sequent calculus $\seqgl$ for $\logicgl$ consists of the rules $\id$, $\negl$, $\negr$, $\disl$, $\disr$, and $\glbox$ shown in \fig~\ref{fig:seq-calc} and is an equivalent variant of the Gentzen calculi $\mathsf{GLSC}$ and $\mathsf{GLS}$ introduced by Sambin and Valentini for $\logicgl$~\cite{SamVal80,SamVal82}.\footnote{$\seqgl$ differs from Sambin and Valentini's original systems in that multisets are used instead of sets, rules for superfluous logical connectives (e.g., conjunction $\land$ and implication $\rightarrow$) have been omitted as these are definable in terms of other rules, and the weakening rules have been absorbed into $\id$ and $\glbox$.} %The system $\seqgl$ is sound and complete for $\logicgl$, admits syntactic cut-elimination, and possesses a variety of fundamental properties as stated in the theorem below (cf.~\cite{GorRam12b,Sha14}). We note that the $\wk$, $\ctrl$, $\ctrr$, $\cut$, and $\lobr$ rules (mentioned in the theorem below) are displayed in \fig~\ref{fig:seq-calc-struc-rules}.
The system $\seqgl$ is sound and complete for $\logicgl$, admits syntactic cut-elimination, and the weakening and contraction rules $\wk$, $\ctrl$, and $\ctrr$ (shown below) are admissible (cf.~\cite{GorRam12b,Sha14}).
\begin{center}
\begin{tabular}{c c c c}
\AxiomC{$\Gamma \sar \Delta$}
\RightLabel{$\wk$}
\UnaryInfC{$\Gamma, \Sigma \sar \Pi, \Delta$}
\DisplayProof

&

\AxiomC{$\Gamma, \phi, \phi \sar \Delta$}
\RightLabel{$\ctrl$}
\UnaryInfC{$\Gamma, \phi \sar \Delta$}
\DisplayProof

&

\AxiomC{$\Gamma \sar \phi, \phi, \Delta$}
\RightLabel{$\ctrr$}
\UnaryInfC{$\Gamma \sar \phi, \Delta$}
\DisplayProof

&

\AxiomC{$\Gamma \sar \phi, \Delta$}
\AxiomC{$\Gamma, \phi \sar \Delta$}
\RightLabel{$\cut$}
\BinaryInfC{$\Gamma \sar \Delta$}
\DisplayProof
\end{tabular}
\end{center}

Shamkanov~\cite{Sha14} showed that equivalent non-wellfounded and cyclic sequent systems could be obtained for $\logicgl$ by taking the sequent calculus for the modal logic $\mathsf{K4}$ and generalizing the notion of proof. The sequent calculus $\mathsf{K4_{seq}}$ is obtained by replacing the $\glbox$ rule in $\seqgl$ with the $\fourbox$ rule shown in \fig~\ref{fig:seq-calc}. Let us now recall Shamkanov's non-wellfounded sequent calculus $\seqgli$ and cyclic sequent calculus $\seqglc$ for $\logicgl$. We present Shamkanov's systems in a \emph{two-sided format}, i.e., using two-sided sequents $\Gamma \sar \Delta$ rather than one-sided sequents of the form $\Gamma$. This makes the correspondence between Shamkanov's systems and $\seqgl$ clearer as well as saves us from having to introduce a new language for $\logicgl$ since one-sided sequents use formulae in negation normal form. Translating proofs with two-sided sequents to proofs with one-sided sequents and vice-versa can be easily obtained by standard techniques, and so, this minor modification causes no problems.

A \emph{derivation} of a sequent $\Gamma \sar \Delta$ is defined to be a (potentially infinite) tree whose nodes are labeled with sequents such that (1) $\Gamma \sar \Delta$ is the root of the tree, and (2) each parent node is taken to be the conclusion of a rule in $\mathsf{K4_{seq}}$ with its children the corresponding premises. A \emph{non-wellfounded proof} is a derivation %that only contains finitely many distinct subtrees. 
such that all leaves are initial sequents. $\seqgli$ is the non-wellfounded sequent system obtained by letting the set of provable sequents be determined by non-wellfounded proofs.

A \emph{cyclic derivation} is a pair $\prf = (\kappa,c)$ such that $\kappa$ is a finite derivation in $\mathsf{K4_{seq}}$ and $c$ is a function with the following properties: (1) $c$ is defined on a subset of the leaves of $\kappa$, (2) the image $c(x)$ lies on the path from the root of $\kappa$ to $x$ and does not coincide with $x$, and (3) both $x$ and $c(x)$ are labeled by the same sequent. If the function $c$ is defined at a leaf $x$, then we say that a \emph{back-link} exists from $x$ to $c(x)$. %An \emph{assumption} in a cyclic derivation $\prf = (\kappa,c)$ is defined to be a leaf of $\kappa$ that is not an initial sequent and is not connected by a back-link from $c$. 
A \emph{cyclic proof} is a cyclic derivation $\prf = (\kappa,c)$ such that every leaf $x$ is labeled by an instance of $\id$ or there exists a back-link from $x$ to the node $c(x)$. $\seqglc$ is the cyclic sequent system obtained by letting the set of provable sequents be determined by cyclic proofs.

%Shamkanov established a three-way correspondence between $\seqgl$, $\seqgli$, and $\seqglc$ by means of three proof mappings.  In particular, Shamkanov showed that (1) every proof in $\seqgl$ can be transformed into a non-wellfounded proof in $\seqgli$, (2) every non-wellfounded proof in $\seqgli$ can be transformed into a cyclic proof in $\seqglc$, and (3) if a cyclic proof of a sequent $\Gamma \sar \Delta$ exists in $\seqglc$, then a proof of $\Gamma \sar \Delta$ exists in $\seqgl$.\footnote{The proof transformation in (1) relies on the admissibility of the $\cut$ rule in $\seqgl$. This is not problematic however since $\seqgl$ admits syntactic cut-elimination.} We remark that Shamkanov did not define a direct proof transformation from $\seqglc$ to $\seqgl$ in step (3), but rather, took a detour through the axiomatization of $\logicgl$ and relied on the completeness of $\seqgl$ to establish the result (see~\cite[Lemma 3.7]{Sha14}). Nevertheless, Shamkanov's proof of (3) can be adjusted to provide a direct proof transformation from $\seqglc$ to $\seqgl$ (see Appendix~\ref{app:B}), yielding syntatic transformations between $\seqgl$, $\seqgli$, and $\seqglc$.

Shamkanov established a three-way correspondence between $\seqgl$, $\seqgli$, and $\seqglc$, providing syntactic transformations mapping proofs between the three systems.\footnote{We note that Shamkanov's proof transformation from $\seqgl$ to $\seqgli$ relies on the admissibility of the $\cut$ rule in $\seqgl$. This is not problematic however since $\seqgl$ admits syntactic cut-elimination.}

%[Theorem 3.2] of Sha14
\begin{theorem}[\cite{Sha14}] $\Gamma \sar \Delta$ is provable in $\seqgl$ \iffi $\Gamma \sar \Delta$  is provable in $\seqgli$ \iffi $\Gamma \sar \Delta$  is provable in $\seqglc$.
\end{theorem}

\subsection{Completing the Correspondences}

\begin{theorem}\label{thm:lin-to-seq}
If $\prf$ is a normal form proof of $\sar \phi$ in $\lingl$, then $\prf$ can be transformed into a proof of $\sar \phi$ in $\seqgl$.
\end{theorem}

\begin{proof} We show how to transform the normal form proof $\prf$ of $\sar \phi$ in $\lingl$ into a proof $\prf'$ of $\sar \phi$ in $\seqgl$ in a bottom-up manner. For the conclusion $\sar \phi$ of the proof $\prf$, we take $\sar \phi$ to be the conclusion of $\prf'$. We now make a case distinction on bottom-up applications of rules applied in $\prf$. For each rule $\negl$, $\negr$, $\disl$, or $\disr$, we translate each premise of the rule as its end component. For example, the $\disl$ rule will be translated as shown below.
\begin{center}
\begin{tabular}{c c}
\AxiomC{$\lnsg \sep \Gamma, \phi \sar \Delta$}
\AxiomC{$\lnsg \sep \Gamma, \psi \sar \Delta$}
\RightLabel{$\disl$}
\BinaryInfC{$\lnsg \sep \Gamma, \phi \lor \psi \sar \Delta$}
\DisplayProof

&

\AxiomC{$\Gamma, \phi \sar \Delta$}
\AxiomC{$\Gamma, \psi \sar \Delta$}
\RightLabel{$\disl$}
\BinaryInfC{$\Gamma, \phi \lor \psi \sar \Delta$}
\DisplayProof
\end{tabular}
\end{center}
Suppose now that we encounter a $\boxrt$ rule while bottom-up translating the proof $\prf$ into a proof in $\seqgl$. Since $\prf$ is in normal form, we know that $\boxrt$ is preceded by a modal block (see \cref{def:normal-form}), that is, $\boxrt$ is (bottom-up) preceded by a block $\block_{\fourru}$ of $\fourru$ rules, which is preceded by a block $\block_{\boxlt}$ of $\boxlt$ rules, i.e., the modal block has the shape shown below. We suppose that $\Box \Sigma_{1}$ are the principal formulae of the $\fourru$ applications, $\Box \Sigma_{2}$ are those formulae principal in both $\fourru$ and $\boxl$ applications, and $\Box \Sigma_{3}$ are those formulae principal only in $\boxl$ applications.
\begin{center}
\AxiomC{$\lnsg \sep \Gamma, \Box \Sigma_{1}, \Box \Sigma_{2}, \Box \Sigma_{3} \sar \Delta \sep \Box \Sigma_{1}, \Box \Sigma_{2}, \Sigma_{2}, \Sigma_{3}, \Box \phi \sar \phi$}
\RightLabel{$\block_{\boxlt}$}
\doubleLine
\UnaryInfC{$\lnsg \sep \Gamma, \Box \Sigma_{1}, \Box \Sigma_{2}, \Box \Sigma_{3} \sar \Delta \sep \Box \Sigma_{1}, \Box \Sigma_{2}, \Box \phi \sar \phi$}
\RightLabel{$\block_{\fourru}$}
\doubleLine
\UnaryInfC{$\lnsg \sep \Gamma, \Box \Sigma_{1}, \Box \Sigma_{2}, \Box \Sigma_{3} \sar \Delta \sep \Box \phi \sar \phi$}
\RightLabel{$\boxr$}
\UnaryInfC{$\lnsg \sep \Gamma, \Box \Sigma_{1}, \Box \Sigma_{2}, \Box \Sigma_{3} \sar \Box \phi, \Delta$}
\DisplayProof
\end{center}
We bottom-up translate the entire block as shown below, where the conclusion is obtained from the end component of the modal block's conclusion. Note that we may apply the $\wk$ rule because it is admissible in $\seqgl$.
\begin{center}
\AxiomC{$\Box \Sigma_{1}, \Box \Sigma_{2}, \Sigma_{2}, \Sigma_{3}, \Box \phi \sar \phi$}
\RightLabel{$\wk$}
%\doubleLine
\UnaryInfC{$\Box \Sigma_{1}, \Box \Sigma_{2}, \Box \Sigma_{3}, \Sigma_{1}, \Sigma_{2}, \Sigma_{3}, \Box \phi \sar \phi$}
\RightLabel{$\glbox$}
\UnaryInfC{$\Gamma, \Box \Sigma_{1}, \Box \Sigma_{2}, \Box \Sigma_{3} \sar \Box \phi, \Delta$}
\DisplayProof
\end{center}
Last, suppose an instance of $\id$ is reached in the translation, as shown below left.
\begin{center}
\begin{tabular}{c c}
\AxiomC{}
\RightLabel{$\id$}
\UnaryInfC{$\lnsg \sep \Gamma, p \sar p, \Delta$}
\DisplayProof

&

\AxiomC{}
\RightLabel{$\id$}
\UnaryInfC{$\Gamma, p \sar p, \Delta$}
\DisplayProof
\end{tabular}
\end{center}
We translate the linear nested sequent as its end component, yielding the Gentzen sequent shown above right, which is an instance of $\id$.
\end{proof}

Last, the following theorem completes the circuit of proof transformations and establishes syntactic correspondences between $\gtgl$, $\treegl$, $\lingl$, $\seqgl$, $\seqgli$, and $\seqglc$. %The theorem is proven in a straightforward manner by induction on the height of a proof in $\seqgl$.  %Note that for a multiset $\Gamma := \phi_{1}, \ldots, \phi_{n}$, we define $x : \Gamma := x : \phi_{1}, \ldots, x : \phi_{n}$.

\begin{theorem}\label{thm:seq-to-lab}
If $\Gamma \sar \Delta$ is provable in $\seqgl$, then $x : \Gamma \sar x : \Delta$ is provable in $\gtgl$.
\end{theorem}

\begin{proof} By induction on the height of the proof $\prf$ in $\seqgl$. The base case immediately follows from Theorem~\ref{thm:lab-seq-properties}-(1), and the $\negl$, $\negr$, $\disl$, and $\disr$ cases of the inductive step straightforwardly follow by applying IH and then the corresponding rule in $\gtgl$. Therefore, we need only show the case where $\prf$ ends with an application of $\glbox$, as shown below left.
%Suppose $\prf$ ends with an application of $\glbox$ as shown below left. 
\begin{center}
\begin{tabular}{c c}
\AxiomC{$\Box \Gamma, \Gamma, \Box \phi \sar \phi$}
\RightLabel{$\glbox$}
\UnaryInfC{$\Sigma, \Box \Gamma \sar \Box \phi, \Delta$}
\DisplayProof

&

\AxiomC{$x : \Box \Gamma, x : \Gamma, x : \Box \phi \sar x : \phi$}
\RightLabel{$\wk$}
\UnaryInfC{$yRx, y : \Box \Gamma, x : \Box \Gamma, x : \Gamma, x : \Box \phi \sar x : \phi$}
\RightLabel{$\boxl$}
\UnaryInfC{$yRx, y : \Box \Gamma, x : \Box \Gamma, x : \Box \phi \sar x : \phi$}
\RightLabel{$\fourru$}
\UnaryInfC{$yRx, y : \Box \Gamma, x : \Box \phi \sar x : \phi$}
\RightLabel{$\boxr$}
\UnaryInfC{$y : \Box \Gamma \sar y : \Box \phi$}
\RightLabel{$\lsub$}
\UnaryInfC{$x : \Box \Gamma \sar x : \Box \phi$}
\RightLabel{$\wk$}
\UnaryInfC{$x : \Sigma, x : \Box \Gamma \sar x : \Box \phi, x : \Delta$}
\DisplayProof
\end{tabular}
\end{center}
To obtain the desired proof, we first apply the hp-admissible $\wk$ rule (\cref{thm:lab-seq-properties}), followed by applications of the $\boxl$ rule and admissible $\fourru$ rule (cf.~\cite{GorRam12}). Applying the $\boxr$ rule, followed by applications of the hp-admissible $\lsub$ and $\wk$ rules (\cref{thm:lab-seq-properties}), gives the desired conclusion.
\end{proof}

%% file: conclusion.tex
%(1) discuss investigate properties of LNS system, (2) look at how to minimize structure in other cyclic systems or if int GL can be turned into sequent system}

%We have solved an open problem concerning the translatability of proofs between various sequent systems for Gödel-Löb provability logic. To do so, we showed how to restructure proofs in $\treegl$ into end

There are various avenues for future research: first, it would be interesting to look into the properties of the new linear nested sequent calculus $\lingl$, investigating additional admissible structural rules, how the system can be amended to allow for the hp-invertibility of all rules, and also looking into syntactic cut-elimination. %Second, we aim to see if the linearization technique provided here can be used to extract new (cut-free) linear nested sequent systems for other modal, intuitionistic, and related logics.  
Second, by employing a methodology for extracting nested sequent systems from relational semantics~\cite{LyoOst24}, we can integrate this approach with the linearization technique to develop a general method for extracting (cut-free) linear nested systems from the semantics of various modal, intuitionistic, and related logics. Third, it seems worthwhile to see if the proof transformation techniques discussed in this paper can be applied to structural cyclic systems (e.g., cyclic labeled sequent systems for classical and intuitionistic Gödel-Löb logic~\cite{DasGieMar24}) to remove extraneous structure and extract simpler (cyclic) Gentzen systems.

%Second, it seems worthwhile to see if the linearization technique provided here can be used to extract new (cut-free) linear nested sequent systems for other modal, intuitionistic, and related logics. 
%Second, since a general methodology for extracting nested sequent systems from labeled sequent systems was recently developed~\cite{LyoOst24}, we could patch this method together with the linearization technique to obtain a general methodology for the construction of (cut-free) linear nested systems for other modal, intuitionistic, and related logics.

%1. properties of lns system
%2. more lns systems for modal and related logics
%3. extract gentzen systems for logics
%3. apply these mtheods to other systems and cyclic systems

%% file: appendix.tex
\section{Proofs for \cref{sec:linearizing}}

\begin{customlem}{\ref{lem:permutations}} The following permutations hold in $\treegl$: 
\begin{description}

\item[(1)] If $\ru$ is a non-end-active local rule and $\ruprime$ is non-initial and end-active, then $\ru$ permutes below $\ruprime$ and $\ruprime$ remains end-active after this permutation;

\item[(2)] if $\ru$ is a non-end-active propagation rule and $\ruprime$ is non-initial and end-active, then $\ru$ permutes below $\ruprime$ and $\ruprime$ remains end-active after this permutation.

%\item[(3)] if a $\boxrt$ rule is non-end-active and $\ru$ is non-initial and end-active, then $\boxrt$ permutes below $\ru$ and $\ru$ remains end-active after this permutation.

%\item[(3)] if a $\boxrt$ rule is non-end-active and $\ru$ is a non-initial, end-active rule, then $\boxrt \permbelow \ru$;

%\item[(3)] if $\ru$ is an end-active propagation rule and $\ruprime$ is an end-active local rule, then $\ru \permbelow \ruprime$;

%\item[(4)] if $\fourru$ rule is end-active and $\boxlt$ is end-active, then $\fourru \permbelow \boxlt$;

%\item[(5)] if a rule $\ru$ is independent from a rule $\ru'$, then $\ru \permbelow \ru'$.
    
\end{description}
\end{customlem}

\begin{proof} Follows from Lemmas~\ref{lem:app-1} and \ref{lem:app-2} below. %Follows from Lemmas~\ref{lem:app-1}--\ref{lem:app-last} below.
\end{proof}

\begin{lemma}\label{lem:app-1}
%If $\ru$ is a non-end-active local rule and $\ruprime$ is non-initial and end-active, then $\ru \permbelow \ruprime$.
If $\ru$ is a non-end-active local rule and $\ruprime$ is non-initial and end-active, then $\ru$ permutes below $\ruprime$ and $\ruprime$ remains end-active after this permutation.
\end{lemma}

\begin{proof} We let $\ru$ be an instance of $\negr$ as the cases where $\ru$ is either $\negl$, $\disl$, or $\disr$ are shown similarly. We show that $\ru$ can be permuted down $\ruprime$ and consider a representative number of cases when $\ruprime$ is either $\disr$, $\fourru$, or $\boxrt$ as the remaining cases are similar.

\medskip

\noindent
$\disr$. By our assumption that $\negr$ is non-end-active and $\disr$ is end-active, we know that the labels $x$ and $y$ are distinct. Hence, we can permute the $\negr$ instance below the $\disr$ instance. Observe that $\disr$ remains end-active after the permutation.
\begin{center}
\begin{tabular}{c c}
\AxiomC{$\tel, \Gamma, x : \phi \sar y : \psi, y : \chi, \Delta$}
\RightLabel{$\negr$}
\UnaryInfC{$\tel, \Gamma \sar x : \neg \phi, y : \psi, y : \chi, \Delta$}
\RightLabel{$\disr$}
\UnaryInfC{$\tel, \Gamma \sar x : \neg \phi, y : \psi \lor \chi, \Delta$}
\DisplayProof

&

\AxiomC{$\tel, \Gamma, x : \phi \sar y : \psi, y : \chi, \Delta$}
\RightLabel{$\disr$}
\UnaryInfC{$\tel, \Gamma, x : \phi \sar y : \psi \lor \chi, \Delta$}
\RightLabel{$\negr$}
\UnaryInfC{$\tel, \Gamma \sar x : \neg \phi, y : \psi \lor \chi, \Delta$}
\DisplayProof
\end{tabular}
\end{center}

\medskip

\noindent
$\fourru$. By our assumption, we know that $z$ is distinct from $y$ in the inferences shown below left, meaning, we can permute $\negr$ below $\fourru$ as shown below right. Observe that $\fourru$ remains end-active after the permutation.
\begin{center}
\begin{tabular}{c c}
\AxiomC{$\tel, xRy, \Gamma, x : \Box \psi, y : \Box \psi, z : \phi \sar \Delta$}
\RightLabel{$\negr$}
\UnaryInfC{$\tel, xRy, \Gamma, x : \Box \psi, y : \Box \psi \sar z : \neg \phi, \Delta$}
\RightLabel{$\fourru$}
\UnaryInfC{$\tel, xRy, \Gamma, x : \Box \psi \sar z : \neg \phi, \Delta$}
\DisplayProof

&

\AxiomC{$\tel, xRy, \Gamma, x : \Box \psi, y : \Box \psi, z : \phi \sar \Delta$}
\RightLabel{$\fourru$}
\UnaryInfC{$\tel, xRy, \Gamma, x : \Box \psi, z : \phi \sar \Delta$}
\RightLabel{$\negr$}
\UnaryInfC{$\tel, xRy, \Gamma, x : \Box \psi \sar z : \neg \phi, \Delta$}
\DisplayProof
\end{tabular}
\end{center}

\medskip

\noindent
$\boxrt$. By our assumption, we know that $z$ is distinct from $y$ in the inferences shown below left, meaning, we can permute $\negr$ below $\boxr$ as shown below right. Trivially, the $\boxr$ rule remains end-active after the permutation.
\begin{center}
\begin{tabular}{c c}
\AxiomC{$\tel, xRy, \Gamma, y : \Box \psi, z : \phi \sar y : \psi, \Delta$}
\RightLabel{$\negr$}
\UnaryInfC{$\tel, xRy, \Gamma, y : \Box \psi \sar y : \psi, z : \neg \phi, \Delta$}
\RightLabel{$\boxrt$}
\UnaryInfC{$\tel, \Gamma \sar x : \Box \psi, z : \neg \phi, \Delta$}
\DisplayProof

&

\AxiomC{$\tel, xRy, \Gamma, y : \Box \psi, z : \phi \sar y : \psi, \Delta$}
\RightLabel{$\boxrt$}
\UnaryInfC{$\tel, \Gamma, z : \phi \sar x : \Box \psi, \Delta$}
\RightLabel{$\negr$}
\UnaryInfC{$\tel, \Gamma \sar x : \Box \psi, z : \neg \phi, \Delta$}
\DisplayProof
\end{tabular}
\end{center}
\end{proof}

\begin{lemma}\label{lem:app-2}
%If $\ru$ is a non-end-active propagation rule and $\ruprime$ is non-initial and end-active, then $\ru \permbelow \ruprime$.
If $\ru$ is a non-end-active propagation rule and $\ruprime$ is non-initial and end-active, then $\ru$ permutes below $\ruprime$ and $\ruprime$ remains end-active after this permutation.
\end{lemma}

\begin{proof} We consider the case where $\ru$ is an instance of $\boxlt$ as the $\fourru$ case is similar. We show that $\ru$ can be permuted down $\ruprime$ and consider a representative number of cases when $\ruprime$ is either $\negl$, $\fourru$, or $\boxrt$ as the remaining cases are similar.

\medskip

\noindent
$\negl$. By our assumption, we know that $z$ is distinct from $y$ in the inferences below left. We can therefore permute $\boxlt$ below $\negl$ as shown below right and we observe that $\negl$ remains end-active.
\begin{center}
\begin{tabular}{c c}
%\AxiomC{$\tel, xRy, \Gamma, x : \Box \phi \sar y : \Box \phi, z : \psi, \Delta$}
\AxiomC{$\tel, xRy, \Gamma, x : \Box \phi, y : \phi \sar z : \psi, \Delta$}
\RightLabel{$\boxlt$}
\UnaryInfC{$\tel, xRy, \Gamma, x : \Box \phi \sar z : \psi, \Delta$}
\RightLabel{$\negl$}
\UnaryInfC{$\tel, xRy, \Gamma, x : \Box \phi, z : \neg \psi \sar \Delta$}
\DisplayProof

&

%\AxiomC{$\tel, xRy, \Gamma, x : \Box \phi \sar y : \Box \phi, z : \psi, \Delta$}
%\RightLabel{$\negl$}
%\UnaryInfC{$\tel, xRy, \Gamma, x : \Box \phi, z : \neg \psi \sar y : \Box \phi, \Delta$}
\AxiomC{$\tel, xRy, \Gamma, x : \Box \phi, y : \phi \sar z : \psi, \Delta$}
\RightLabel{$\negl$}
\UnaryInfC{$\tel, xRy, \Gamma, x : \Box \phi, y : \phi, z : \neg \psi \sar \Delta$}
\RightLabel{$\boxlt$}
\UnaryInfC{$\tel, xRy, \Gamma, x : \Box \phi, z : \neg \psi \sar \Delta$}
\DisplayProof
\end{tabular}
\end{center}

\medskip

\noindent
$\fourru$. Let us suppose we have a $\boxlt$ instance followed by a $\fourru$ instance. There are two cases to consider: either the principal formula of $\boxlt$ is the same as for $\fourru$, or the principal formulae are distinct. We show the first case as the second case is similar. Then, our inferences are of the form shown below left, where $y$ and $z$ are distinct due to our assumption. We may permute $\boxlt$ below $\fourru$ as shown below right and we observe that $\fourru$ remains end-active.
\begin{center}
\begin{tabular}{c c}
%\AxiomC{$\tel, xRy, xRz, \Gamma, x : \Box \phi,  z : \Box \phi \sar y : \Box \phi, \Delta$}
\AxiomC{$\tel, xRy, xRz, \Gamma, x : \Box \phi, y : \phi, z : \Box \phi \sar \Delta$}
\RightLabel{$\boxlt$}
\UnaryInfC{$\tel, xRy, xRz, \Gamma, x : \Box \phi, z : \Box \phi \sar \Delta$}
\RightLabel{$\fourru$}
\UnaryInfC{$\tel, xRy, xRz, \Gamma, x : \Box \phi \sar \Delta$}
\DisplayProof

&

\AxiomC{$\tel, xRy, xRz, \Gamma, x : \Box \phi, y : \phi, z : \Box \phi \sar \Delta$}
\RightLabel{$\fourru$}
\UnaryInfC{$\tel, xRy, xRz, \Gamma, x : \Box \phi, y : \phi \sar \Delta$}
\RightLabel{$\boxlt$}
\UnaryInfC{$\tel, xRy, xRz, \Gamma, x : \Box \phi \sar \Delta$}
\DisplayProof
\end{tabular}
\end{center}

\medskip

\noindent
$\boxrt$. Suppose we have an instance of $\boxlt$ followed by an application of $\boxrt$ as shown below.
\begin{center}
%\AxiomC{$\tseq$}
\AxiomC{$\tel, xRy, zRu, \Gamma, x : \Box \phi, y : \phi, u : \Box \psi  \sar u : \psi, \Delta$}
\RightLabel{$\boxlt$}
\UnaryInfC{$\tel, xRy, zRu, \Gamma, x : \Box \phi, u : \Box \psi \sar u : \psi, \Delta$}
\RightLabel{$\boxrt$}
\UnaryInfC{$\tel, xRy, \Gamma, x : \Box \phi \sar z : \Box \psi, \Delta$}
\DisplayProof
\end{center}
By our assumption, the labels $y$ and $u$ are distinct, meaning, we can permute $\boxlt$ below $\boxrt$ as shown below. Trivially, $\boxr$ remains end-active after the permutation is performed.
\begin{center}
\AxiomC{$\tel, xRy, zRu, \Gamma, x : \Box \phi, y : \phi, u : \Box \psi  \sar u : \psi, \Delta$}
\RightLabel{$\boxrt$}
\UnaryInfC{$\tel, xRy, \Gamma, x : \Box \phi, y : \phi \sar z : \Box \psi, \Delta$}
\RightLabel{$\boxlt$}
\UnaryInfC{$\tel, xRy, \Gamma, x : \Box \phi \sar z : \Box \psi, \Delta$}
\DisplayProof
\end{center}
\end{proof}

\begin{customthm}{\ref{thm:tree-to-lin}}
%If $\prf$ is an end-active proof of a sequent $\sar x : \phi$ in $\treegl$, then $\prf$ can be transformed into a proof in $\lingl$.
Each end-active proof in $\treegl$ can be transformed into a proof in $\lingl$.
\end{customthm}

%applab (change phrasing below in arXiv version)
\begin{proof} We have included additional cases of the inductive step that are not included in the main text. 

\medskip

\noindent
$\fourru$. Let us suppose that $\prf$ ends with an application of $\fourru$ as shown  below.
\begin{center}
\AxiomC{$\tel, xRy, \Gamma, x : \Box \phi, y : \Box \phi \sar \Delta$}
\RightLabel{$\fourru$}
\UnaryInfC{$\tel, xRy, \Gamma, x : \Box \phi \sar \Delta$}
\DisplayProof
\end{center}
By IH, we know there exists a path $y_{1}, \ldots, y_{n}$ of labels from the root $y_{1}$ to the leaf $y_{n}$ in the premise of $\fourru$ such that $\lnsg = \Gamma(y_{1}) \sar \Delta(y_{1}) \sep \cdots \sep \Gamma(y_{n}) \sar \Delta(y_{n})$ is provable in $\lingl$. Since $\fourru$ is end-active, there are three cases to consider: either (1) neither $x$ nor $y$ occur along the path, (2) only $x$ occurs along the path, or (3) both $x$ and $y$ occur along the path. In each case, the conclusion is obtained by taking the linear nested sequent corresponding to the path $y_{1}, \ldots, y_{n}$ in the conclusion of the $\fourru$ instance above. In the first and second cases, we translate the entire $\fourru$ instance as the single linear nested sequent $\lnsg$. %In the second case, we translate the $\fourru$ instance as the linear nested sequent shown below, where $x = y_{n-1}$ and $\Gamma(y_{n-1}) := \Sigma, \Box \phi$.
%$$\Gamma(y_{1}) \sar \Delta(y_{1}) \sep \cdots \sep \Sigma, \Box \phi \sar \Delta(y_{n-1}) \sep \Gamma(y_{n}) \sar \Delta(y_{n})
%$$
In the third case, we have that $x = y_{n-1}$ and $y = y_{n}$, meaning, the premise of the $\fourru$ instance shown below is provable in $\lingl$ by IH, where $\Gamma(y_{n-1}) = \Sigma_{1}, \Box \phi$ and $\Gamma(y_{n}) = \Sigma_{2}, \Box \phi$. As shown below, a single application of $\fourru$ yields the desired conclusion.
\begin{center}
\AxiomC{$\Gamma(y_{1}) \sar \Delta(y_{1}) \sep \cdots \sep \Sigma_{1}, \Box \phi \sar \Delta(y_{n-1}) \sep \Sigma_{2}, \Box \phi \sar \Delta(y_{n})$}
\RightLabel{$\fourru$}
\UnaryInfC{$\Gamma(y_{1}) \sar \Delta(y_{1}) \sep \cdots \sep \Sigma_{1}, \Box \phi \sar \Delta(y_{n-1}) \sep \Sigma_{2} \sar \Delta(y_{n})$}
\DisplayProof
\end{center}

\medskip

\noindent
$\disl$. Let us suppose that $\prf$ ends with an instance of $\disl$ as shown below.
\begin{center}
\AxiomC{$\tel, \Gamma, x : \phi \sar \Delta$}
\AxiomC{$\tel, \Gamma, x : \psi \sar \Delta$}
\RightLabel{$\disl$}
\BinaryInfC{$\tel, \Gamma, x : \phi \lor \psi \sar \Delta$}
\DisplayProof
\end{center}
By IH, we know there exist paths $v = y_{1}, \ldots, y_{n}$ and $v = z_{1}, \ldots, z_{k}$ of labels from the root $v$ to the leaves $y_{n}$ and $z_{k}$ in the premises of $\disl$ such that $\lnsg = \Gamma(v) \sar \Delta(v) \sep \cdots \sep \Gamma(y_{n}) \sar \Delta(y_{n})$ and $\lnsh = \Gamma(v) \sar \Delta(v) \sep \cdots \sep \Gamma(z_{k}) \sar \Delta(z_{k})$ are provable in $\lingl$. There are two cases to consider: either (1) $x \neq y_{n}$ or $x \neq z_{k}$, or (2) $x = y_{n} = z_{k}$. In the first case, if $x \neq y_{n}$, then we translate the entire $\disl$ inference as the single linear nested sequent $\lnsg$, and if $x \neq z_{k}$, then we translate the entire $\disl$ inference as $\lnsh$. In the second case, we know that the left premise $\lnsg$ and right premise $\lnsh$ of the $\disl$ inference below are provable with $\Gamma(y_{n}) = \Sigma, \phi$ and $\Gamma(z_{k}) = \Sigma, \psi$, and so, a single application of $\disl$ gives the desired result.
\begin{center}
\AxiomC{$\Gamma(y_{1}) \sar \Delta(y_{1}) \sep \cdots \sep \Sigma, \phi \sar \Delta(y_{n})$}
\AxiomC{$\Gamma(y_{1}) \sar \Delta(y_{1}) \sep \cdots \sep \Sigma, \psi \sar  \Delta(y_{n})$}
\RightLabel{$\disl$}
\BinaryInfC{$\Gamma(y_{1}) \sar \Delta(y_{1}) \sep \cdots \sep \Sigma, \phi \lor \psi \sar \Delta(y_{n})$}
\DisplayProof
\end{center}
\end{proof}

\begin{customthm}{\ref{thm:lingl-normalizable}}
Every proof in $\lingl$ can be transformed into a proof in normal form.
\end{customthm}

\begin{proof} Let $\prf$ be a proof in $\lingl$. We consider an arbitrary instance of a $\boxrt$ rule in $\prf$ and first show that every length-consistent $\fourru$ rule above $\boxrt$ can be permuted down into a block $\block$ of $\fourru$ rules above $\boxrt$. Afterward, we will show that every length-consistent $\boxlt$ rule can be permuted down into a block $\block'$ of $\boxlt$ rules above $\block$. As a result, all length-consistent local rules will occur in a block $\block''$ above the premise of the block $\block'$, showing that $\boxrt$ is preceded by a complete block. As these permutations can be performed for every $\boxrt$ instance in $\prf$, we obtain a normal form proof as the result.

Let us choose an application of $\boxrt$ in $\prf$, as shown below, preceded by a (potentially empty) block $\blockii$ of $\fourru$ rules.
\begin{center}
\AxiomC{$\vdots$}
\RightLabel{$\blockii$}
\doubleLine
\UnaryInfC{$\lnsg \sep \Gamma \sar \Delta \sep \Box \phi \sar \phi$}
\RightLabel{$\boxrt$}
\UnaryInfC{$\lnsg \sep \Gamma \sar \Box \phi, \Delta$}
\DisplayProof
\end{center}
We now select a bottom-most, length-consistent application of a $\fourru$ rule above the chosen $\boxrt$ application that does not occur within the block $\blockii$ of $\fourru$ rules. We show that $\fourru$ can be permuted below every local rule and $\boxlt$ rule until it reaches and joins the $\blockii$ block. We show that $\fourru$ can be permuted below $\negr$, $\disl$, and $\boxlt$ as the remaining cases are similar. Note that we are guaranteed that no other $\boxrt$ applications occur below $\fourru$ and above $\blockii$ since then $\fourru$ would not be length-consistent with the chosen $\boxrt$ application.

Suppose $\fourru$ occurs above a $\negr$ application as shown below left. The rules can be permuted as shown below right.
\begin{center}
\begin{tabular}{c c}
\AxiomC{$\lnsg \sep \Gamma, \Box \phi \sar \Delta \sep \Sigma, \psi, \Box \phi \sar \Pi$}
\RightLabel{$\fourru$}
\UnaryInfC{$\lnsg \sep \Gamma, \Box \phi \sar \Delta \sep \Sigma, \psi \sar \Pi$}
\RightLabel{$\negr$}
\UnaryInfC{$\lnsg \sep \Gamma, \Box \phi \sar \Delta \sep \Sigma \sar \neg \psi, \Pi$}
\DisplayProof

&

\AxiomC{$\lnsg \sep \Gamma, \Box \phi \sar \Delta \sep \Sigma, \psi, \Box \phi \sar \Pi$}
\RightLabel{$\negr$}
\UnaryInfC{$\lnsg \sep \Gamma, \Box \phi \sar \Delta \sep \Sigma, \Box \phi \sar \neg \psi, \Pi$}
\RightLabel{$\fourru$}
\UnaryInfC{$\lnsg \sep \Gamma, \Box \phi \sar \Delta \sep \Sigma \sar \neg \psi, \Pi$}
\DisplayProof
\end{tabular}
\end{center}
Suppose that we have $\fourru$ followed by an application of the $\disl$ rule.
\begin{center}
\AxiomC{$\lnsg \sep \Gamma, \Box \phi \sar \Delta \sep \Sigma, \Box \phi, \psi \sar \Pi$}
\RightLabel{$\fourru$}
\UnaryInfC{$\lnsg \sep \Gamma, \Box \phi \sar \Delta \sep \Sigma, \psi \sar \Pi$}
\AxiomC{$\lnsg \sep \Gamma, \Box \phi \sar \Delta \sep \Sigma, \chi \sar \Pi$}
\RightLabel{$\disl$}
\BinaryInfC{$\lnsg \sep \Gamma, \Box \phi \sar \Delta \sep \Sigma, \psi \lor \chi \sar \Pi$}
\DisplayProof
\end{center}
Invoking the hp-invertibility of $\fourru$ (\cref{lem:hp-invert-lingl}), we can permute $\fourru$ below $\disl$ as shown below.
\begin{center}
\AxiomC{$\lnsg \sep \Gamma, \Box \phi \sar \Delta \sep \Sigma, \Box \phi, \psi \sar \Pi$}
\AxiomC{$\lnsg \sep \Gamma, \Box \phi \sar \Delta \sep \Sigma, \chi \sar \Pi$}
\RightLabel{$\fourru^{-1}$}
\UnaryInfC{$\lnsg \sep \Gamma, \Box \phi \sar \Delta \sep \Sigma, \Box \phi, \chi \sar \Pi$}
\RightLabel{$\disl$}
\BinaryInfC{$\lnsg \sep \Gamma, \Box \phi \sar \Delta \sep \Sigma, \Box \phi, \psi \lor \chi \sar \Pi$}
\RightLabel{$\fourru$}
\UnaryInfC{$\lnsg \sep \Gamma, \Box \phi \sar \Delta \sep \Sigma, \psi \lor \chi \sar \Pi$}
\DisplayProof
\end{center}
Last, we show (below left) one of the cases where $\fourru$ is applied above a $\boxlt$ rule. %Using the hp-invertibility of $\fourru$ (\cref{lem:hp-invert-lingl}), 
We can permute the rules as shown below right.
\begin{center}
\begin{tabular}{c c}
%\AxiomC{$\lnsg \sep \Gamma, \Box \psi, \Box \phi \sar \Delta \sep \Sigma \sar \Box \phi, \Pi$}
\AxiomC{$\lnsg \sep \Gamma, \Box \psi, \Box \phi \sar \Delta \sep \Sigma, \Box \psi, \phi \sar \Pi$}
\RightLabel{$\fourru$}
\UnaryInfC{$\lnsg \sep \Gamma, \Box \psi, \Box \phi \sar \Delta \sep \Sigma, \phi \sar \Pi$}
\RightLabel{$\boxl$}
\UnaryInfC{$\lnsg \sep \Gamma, \Box \psi, \Box \phi \sar \Delta \sep \Sigma \sar \Pi$}
\DisplayProof  

&

%\AxiomC{$\lnsg \sep \Gamma, \Box \psi, \Box \phi \sar \Delta \sep \Sigma \sar \Box \phi, \Pi$}
%\RightLabel{$\fourru^{-1}$}
%\UnaryInfC{$\lnsg \sep \Gamma, \Box \psi, \Box \phi \sar \Delta \sep \Sigma, \Box \psi \sar \Box \phi, \Pi$}
\AxiomC{$\lnsg \sep \Gamma, \Box \psi, \Box \phi \sar \Delta \sep \Sigma, \Box \psi, \phi \sar \Pi$}
\RightLabel{$\boxl$}
\UnaryInfC{$\lnsg \sep \Gamma, \Box \psi, \Box \phi \sar \Delta \sep \Sigma, \Box \psi \sar \Pi$}
\RightLabel{$\fourru$}
\UnaryInfC{$\lnsg \sep \Gamma, \Box \psi, \Box \phi \sar \Delta \sep \Sigma \sar \Pi$}
\DisplayProof  
\end{tabular}
\end{center}

We can repeat the above downward permutations of bottom-most, length-consistent $\fourru$ rules, so that all length-consistent $\fourru$ rules occur in a block $\block$ above $\boxrt$ as shown below, where we let $\blockii'$ be a (potentially empty) block of $\boxlt$ rules above the $\block$ block of $\fourru$ rules.
\begin{center}
\AxiomC{$\vdots$}
\RightLabel{$\blockii'$}
\doubleLine
\UnaryInfC{$\lnsh$}
\RightLabel{$\block$}
\doubleLine
\UnaryInfC{$\lnsg \sep \Gamma \sar \Delta \sep \Box \phi \sar \phi$}
\RightLabel{$\boxrt$}
\UnaryInfC{$\lnsg \sep \Gamma \sar \Box \phi, \Delta$}
\DisplayProof
\end{center}
Next we show that every length-consistent $\boxlt$ rule occurring above the block $\blockii'$ can be permuted down to the block $\blockii'$. Let $\boxlt$ occurring above the block $\blockii'$ be length-consistent with the chosen $\boxrt$ rule. Notice that we need only consider downward permutations of $\boxlt$ rules with local rules as all $\fourru$ rules have already been permuted downward and no other $\boxrt$ rule can occur between $\boxlt$ and $\blockii'$ because then $\boxlt$ would not be length-consistent. We show how to permute the $\boxlt$ rule below a $\disl$ instance; the remaining cases are simple and similar.
\begin{center}
\AxiomC{$\lnsg \sep \Gamma, \Box \phi \sar \Delta \sep \Sigma, \psi, \phi \sar \Pi$}
\RightLabel{$\boxl$}
\UnaryInfC{$\lnsg \sep \Gamma, \Box \phi \sar \Delta \sep \Sigma, \psi \sar \Pi$}
\AxiomC{$\lnsg \sep \Gamma, \Box \phi \sar \Delta \sep \Sigma, \chi \sar \Pi$}
\RightLabel{$\disl$}
\BinaryInfC{$\lnsg \sep \Gamma, \Box \phi \sar \Delta \sep \Sigma, \psi \lor \chi \sar \Pi$}
\DisplayProof
\end{center}
By using the hp-invertibility of $\boxlt$ (\cref{lem:hp-invert-lingl}), we can permute $\boxlt$ below the $\disl$ rule.
\begin{center}
\AxiomC{$\lnsg \sep \Gamma, \Box \phi \sar \Delta \sep \Sigma, \phi, \psi \sar \Pi$}
\AxiomC{$\lnsg \sep \Gamma, \Box \phi \sar \Delta \sep \Sigma, \chi \sar \Pi$}
\RightLabel{$\boxlt^{-1}$}
\UnaryInfC{$\lnsg \sep \Gamma, \Box \phi \sar \Delta \sep \Sigma, \phi, \chi \sar \Pi$}
\RightLabel{$\disl$}
\BinaryInfC{$\lnsg \sep \Gamma, \Box \phi \sar \Delta \sep \Sigma, \phi, \psi \lor \chi \sar \Pi$}
\RightLabel{$\boxlt$}
\UnaryInfC{$\lnsg \sep \Gamma, \Box \phi \sar \Delta \sep \Sigma, \psi \lor \chi \sar \Pi$}
\DisplayProof
\end{center}
By successively permuting all $\boxlt$ rules down into a block above $\block$, we have that $\boxrt$ is preceded by a complete block in the proof. As argued above, this implies that every proof in $\lingl$ can be put into normal form.
\end{proof}

%% file: main.bbl
\begin{thebibliography}{10}

\bibitem{AfsLei17}
Bahareh Afshari and Graham~E. Leigh.
\newblock Cut-free completeness for modal mu-calculus.
\newblock In {\em 2017 32nd Annual ACM/IEEE Symposium on Logic in Computer
  Science (LICS)}, pages 1--12, 2017.
\newblock \href {https://doi.org/10.1109/LICS.2017.8005088}
  {\path{doi:10.1109/LICS.2017.8005088}}.

\bibitem{Avr84}
Arnon Avron.
\newblock On modal systems having arithmetical interpretations.
\newblock {\em Journal of Symbolic Logic}, 49(3):935–942, 1984.
\newblock \href {https://doi.org/10.2307/2274147} {\path{doi:10.2307/2274147}}.

\bibitem{Bro05}
James Brotherston.
\newblock Cyclic proofs for first-order logic with inductive definitions.
\newblock In Bernhard Beckert, editor, {\em Automated Reasoning with Analytic
  Tableaux and Related Methods}, pages 78--92, Berlin, Heidelberg, 2005.
  Springer Berlin Heidelberg.

\bibitem{Bro07}
James Brotherston.
\newblock Formalised inductive reasoning in the logic of bunched implications.
\newblock In Hanne~Riis Nielson and Gilberto Fil{\'e}, editors, {\em Static
  Analysis}, pages 87--103, Berlin, Heidelberg, 2007. Springer Berlin
  Heidelberg.

\bibitem{BroSim10}
James Brotherston and Alex Simpson.
\newblock {Sequent calculi for induction and infinite descent}.
\newblock {\em Journal of Logic and Computation}, 21(6):1177--1216, 10 2010.
\newblock \href {https://doi.org/10.1093/logcom/exq052}
  {\path{doi:10.1093/logcom/exq052}}.

\bibitem{Bru09}
Kai Br{\"{u}}nnler.
\newblock Deep sequent systems for modal logic.
\newblock {\em Archive for Mathematical Logic}, 48(6):551--577, 2009.
\newblock \href {https://doi.org/10.1007/s00153-009-0137-3}
  {\path{doi:10.1007/s00153-009-0137-3}}.

\bibitem{Bul92}
Robert~A. Bull.
\newblock Cut elimination for propositional dynamic logic without *.
\newblock {\em Zeitschrift f\"ur Mathematische Logik und Grundlagen der
  Mathematik}, 38(2):85--100, 1992.

\bibitem{CasCerGasHer97}
Marcos~A Castilho, Luis~Farinas del Cerro, Olivier Gasquet, and Andreas Herzig.
\newblock Modal tableaux with propagation rules and structural rules.
\newblock {\em Fundamenta Informaticae}, 32(3, 4):281--297, 1997.

\bibitem{CiaLyoRamTiu21}
Agata Ciabattoni, Tim Lyon, Revantha Ramanayake, and Alwen Tiu.
\newblock Display to labelled proofs and back again for tense logics.
\newblock {\em ACM Transactions on Computational Logic}, 22(3):1–31, 2021.
\newblock \href {https://doi.org/10.1145/3460492} {\path{doi:10.1145/3460492}}.

\bibitem{DasGir23}
Anupam Das and Marianna Girlando.
\newblock Cyclic hypersequent system for transitive closure logic.
\newblock {\em Journal of Automated Reasoning}, 67(3):27, 2023.
\newblock \href {https://doi.org/10.1007/s10817-023-09675-1}
  {\path{doi:10.1007/s10817-023-09675-1}}.

\bibitem{DasGieMar24}
Anupam Das, Iris van~der Giessen, and Sonia Marin.
\newblock {Intuitionistic G\"{o}del-L\"{o}b Logic, \`{a} la Simpson: Labelled
  Systems and Birelational Semantics}.
\newblock In Aniello Murano and Alexandra Silva, editors, {\em 32nd EACSL
  Annual Conference on Computer Science Logic (CSL 2024)}, volume 288 of {\em
  Leibniz International Proceedings in Informatics (LIPIcs)}, pages
  22:1--22:18, Dagstuhl, Germany, 2024. Schloss Dagstuhl -- Leibniz-Zentrum
  f{\"u}r Informatik.
\newblock \href {https://doi.org/10.4230/LIPIcs.CSL.2024.22}
  {\path{doi:10.4230/LIPIcs.CSL.2024.22}}.

\bibitem{Fit72}
Melvin Fitting.
\newblock Tableau methods of proof for modal logics.
\newblock {\em Notre Dame Journal of Formal Logic}, 13(2):237--247, 1972.

\bibitem{Fit14}
Melvin Fitting.
\newblock Nested sequents for intuitionistic logics.
\newblock {\em Notre Dame Journal of Formal Logic}, 55(1):41--61, 2014.

\bibitem{GorPosTiu08}
Rajeev Gor{\'{e}}, Linda Postniece, and Alwen Tiu.
\newblock Cut-elimination and proof-search for bi-intuitionistic logic using
  nested sequents.
\newblock In Carlos Areces and Robert Goldblatt, editors, {\em Advances in
  Modal Logic 7}, pages 43--66. College Publications, 2008.
\newblock URL:
  \url{http://www.aiml.net/volumes/volume7/Gore-Postniece-Tiu.pdf}.

\bibitem{GorRam12}
Rajeev Gor{\'{e}} and Revantha Ramanayake.
\newblock Labelled tree sequents, tree hypersequents and nested (deep)
  sequents.
\newblock In Thomas Bolander, Torben Bra{\"{u}}ner, Silvio Ghilardi, and
  Lawrence~S. Moss, editors, {\em Advances in Modal Logic 9}, pages 279--299.
  College Publications, 2012.
\newblock URL: \url{http://www.aiml.net/volumes/volume9/Gore-Ramanayake.pdf}.

\bibitem{GorRam12b}
Rajeev Gor{\'e} and Revantha Ramanayake.
\newblock Valentini’s cut-elimination for provability logic resolved.
\newblock {\em The Review of Symbolic Logic}, 5(2):212–238, 2012.
\newblock \href {https://doi.org/10.1017/S1755020311000323}
  {\path{doi:10.1017/S1755020311000323}}.

\bibitem{IshKik07}
Ryo Ishigaki and Kentaro Kikuchi.
\newblock Tree-sequent methods for subintuitionistic predicate logics.
\newblock In Nicola Olivetti, editor, {\em Automated Reasoning with Analytic
  Tableaux and Related Methods}, volume 4548 of {\em Lecture Notes in Computer
  Science}, pages 149--164, Berlin, Heidelberg, 2007. Springer Berlin
  Heidelberg.

\bibitem{Kan57}
Stig Kanger.
\newblock {\em Provability in logic}.
\newblock Almqvist \& Wiksell, 1957.

\bibitem{Kas94}
Ryo Kashima.
\newblock Cut-free sequent calculi for some tense logics.
\newblock {\em Studia Logica}, 53(1):119--135, 1994.

\bibitem{KuzLel18}
Roman Kuznets and Bj{\"{o}}rn Lellmann.
\newblock Interpolation for intermediate logics via hyper- and linear nested
  sequents.
\newblock In Guram Bezhanishvili, Giovanna D'Agostino, George Metcalfe, and
  Thomas Studer, editors, {\em Advances in Modal Logic 12}, pages 473--492.
  College Publications, 2018.

\bibitem{Lel15}
Bj{\"o}rn Lellmann.
\newblock Linear nested sequents, 2-sequents and hypersequents.
\newblock In Hans De~Nivelle, editor, {\em Automated Reasoning with Analytic
  Tableaux and Related Methods}, volume 9323 of {\em Lecture Notes in Computer
  Science}, pages 135--150, Cham, 2015. Springer International Publishing.

\bibitem{Lyo21}
Tim Lyon.
\newblock {On the correspondence between nested calculi and semantic systems
  for intuitionistic logics}.
\newblock {\em Journal of Logic and Computation}, 31(1):213--265, 12 2020.
\newblock \href {https://doi.org/10.1093/logcom/exaa078}
  {\path{doi:10.1093/logcom/exaa078}}.

\bibitem{Lyo21thesis}
Tim Lyon.
\newblock {\em Refining labelled systems for modal and constructive logics with
  applications}.
\newblock PhD thesis, TU Wien, 2021.
\newblock URL: \url{https://arxiv.org/abs/2107.14487}.

\bibitem{LyoTiuGorClo20}
Tim Lyon, Alwen Tiu, Rajeev Gor\'{e}, and Ranald Clouston.
\newblock {Syntactic interpolation for tense logics and bi-intuitionistic logic
  via nested sequents}.
\newblock In Maribel Fern\'{a}ndez and Anca Muscholl, editors, {\em 28th EACSL
  Annual Conference on Computer Science Logic (CSL 2020)}, volume 152 of {\em
  Leibniz International Proceedings in Informatics (LIPIcs)}, pages
  28:1--28:16, Dagstuhl, Germany, 2020. Schloss Dagstuhl -- Leibniz-Zentrum
  f{\"u}r Informatik.
\newblock \href {https://doi.org/10.4230/LIPIcs.CSL.2020.28}
  {\path{doi:10.4230/LIPIcs.CSL.2020.28}}.

\bibitem{Lyo21a}
Tim~S. Lyon.
\newblock Nested sequents for intuitionistic modal logics via structural
  refinement.
\newblock In Anupam Das and Sara Negri, editors, {\em Automated Reasoning with
  Analytic Tableaux and Related Methods}, pages 409--427, Cham, 2021. Springer
  International Publishing.

\bibitem{Tic23}
Tim~S. Lyon, Agata Ciabattoni, Didier Galmiche, Dominique Larchey-Wendling,
  Daniel Méry, Nicola Olivetti, and Revantha Ramanayake.
\newblock Internal and external calculi: Ordering the jungle without being lost
  in translations.
\newblock {\em Found on arXiv}, 2023.
\newblock URL: \url{https://arxiv.org/abs/2312.03426}.

\bibitem{LyoOrl23}
Tim~S. Lyon and Eugenio Orlandelli.
\newblock Nested sequents for quantified modal logics.
\newblock In Revantha Ramanayake and Josef Urban, editors, {\em Automated
  Reasoning with Analytic Tableaux and Related Methods}, pages 449--467, Cham,
  2023. Springer Nature Switzerland.

\bibitem{LyoOst24}
Tim~S. Lyon and Piotr Ostropolski-Nalewaja.
\newblock Foundations for an abstract proof theory in the context of horn
  rules.
\newblock {\em Found on arXiv}, 2024.
\newblock URL: \url{https://arxiv.org/abs/2304.05697}.

\bibitem{Mas92}
Andrea Masini.
\newblock 2-sequent calculus: A proof theory of modalities.
\newblock {\em Annals of Pure and Applied Logic}, 58(3):229--246, 1992.

\bibitem{Mas93}
Andrea Masini.
\newblock 2-sequent calculus: Intuitionism and natural deduction.
\newblock {\em Journal of Logic and Computation}, 3(5):533--562, 1993.

\bibitem{Neg05}
Sara Negri.
\newblock Proof analysis in modal logic.
\newblock {\em Journal of Philosophical Logic}, 34(5):507--544, 2005.
\newblock \href {https://doi.org/10.1007/s10992-005-2267-3}
  {\path{doi:10.1007/s10992-005-2267-3}}.

\bibitem{NiwWal96}
Damian Niwiński and Igor Walukiewicz.
\newblock Games for the $\mu$-calculus.
\newblock {\em Theoretical Computer Science}, 163(1):99--116, 1996.
\newblock \href {https://doi.org/10.1016/0304-3975(95)00136-0}
  {\path{doi:10.1016/0304-3975(95)00136-0}}.

\bibitem{PimLelRam19}
Elaine Pimentel, Revantha Ramanayake, and Bj{\"o}rn Lellmann.
\newblock Sequentialising nested systems.
\newblock In Serenella Cerrito and Andrei Popescu, editors, {\em Automated
  Reasoning with Analytic Tableaux and Related Methods}, volume 11714 of {\em
  Lecture Notes in Computer Science}, pages 147--165, Cham, 2019. Springer
  International Publishing.

\bibitem{Pog09}
Francesca Poggiolesi.
\newblock The method of tree-hypersequents for modal propositional logic.
\newblock In David Makinson, Jacek Malinowski, and Heinrich Wansing, editors,
  {\em Towards Mathematical Philosophy}, volume~28 of {\em Trends in logic},
  pages 31--51. Springer, 2009.
\newblock \href {https://doi.org/10.1007/978-1-4020-9084-4\_3}
  {\path{doi:10.1007/978-1-4020-9084-4\_3}}.

\bibitem{Pog09b}
Francesca Poggiolesi.
\newblock A purely syntactic and cut-free sequent calculus for the modal logic
  of provability.
\newblock {\em The Review of Symbolic Logic}, 2(4):593–611, 2009.
\newblock \href {https://doi.org/10.1017/S1755020309990244}
  {\path{doi:10.1017/S1755020309990244}}.

\bibitem{SamVal80}
G.~Sambin and S.~Valentini.
\newblock A modal sequent calculus for a fragment of arithmetic.
\newblock {\em Studia Logica: An International Journal for Symbolic Logic},
  39(2/3):245--256, 1980.
\newblock URL: \url{http://www.jstor.org/stable/20014984}.

\bibitem{SamVal82}
Giovanni Sambin and Silvio Valentini.
\newblock The modal logic of provability. the sequential approach.
\newblock {\em Journal of Philosophical Logic}, 11(3):311--342, 1982.
\newblock URL: \url{http://www.jstor.org/stable/30226252}.

\bibitem{Seg71}
Krister Segerberg.
\newblock {\em An Essay in Classical Modal Logic}.
\newblock Uppsala: Filosofiska Föreningen och Filosofiska Institutionen vid
  Uppsala Universitet, 1971.

\bibitem{Sha14}
D.~S. Shamkanov.
\newblock Circular proofs for the {G{\"o}del-L{\"o}b} provability logic.
\newblock {\em Mathematical Notes}, 96(3):575--585, 2014.
\newblock \href {https://doi.org/10.1134/S0001434614090326}
  {\path{doi:10.1134/S0001434614090326}}.

\bibitem{Sim94}
Alex~K Simpson.
\newblock {\em The proof theory and semantics of intuitionistic modal logic}.
\newblock PhD thesis, University of Edinburgh. College of Science and
  Engineering. School of Informatics, 1994.

\bibitem{Sol76}
Robert~M. Solovay.
\newblock Provability interpretations of modal logic.
\newblock {\em Israel Journal of Mathematics}, 25(3):287--304, 1976.
\newblock \href {https://doi.org/10.1007/BF02757006}
  {\path{doi:10.1007/BF02757006}}.

\bibitem{Str13}
Lutz Stra{\ss}burger.
\newblock Cut elimination in nested sequents for intuitionistic modal logics.
\newblock In Frank Pfenning, editor, {\em Foundations of Software Science and
  Computation Structures}, volume 7794 of {\em Lecture Notes in Computer
  Science}, pages 209--224, Berlin, Heidelberg, 2013. Springer Berlin
  Heidelberg.

\bibitem{Vig00}
Luca Vigan{\`o}.
\newblock {\em Labelled Non-Classical Logics}.
\newblock Springer Science \& Business Media, 2000.

\end{thebibliography}
